\documentclass[11pt,reqno]{amsart}
	\usepackage[T1]{fontenc}
	\usepackage[utf8]{inputenc}
	\usepackage[english]{babel}
	\usepackage[english=american]{csquotes}
	\usepackage{graphicx}
	\usepackage{subfigure}
	\usepackage{float}
	\usepackage{amssymb,amsmath}
	\usepackage{mathtools}
	\usepackage{color}
	\usepackage{upgreek}
	\usepackage{stackrel}
 	\usepackage[small,nohug,heads=LaTeX]{diagrams} 
	\usepackage[pdftex,plainpages=false,pdfborder={0 0 0},pdfpagelabels,unicode=true,breaklinks=true]{hyperref}

\newtheorem{Theorem}{Theorem}[section]
\newtheorem{Lemma}[Theorem]{Lemma}

\newtheorem{Corollary}[Theorem]{Corollary}

\newtheorem{Remark}[Theorem]{Remark}
\newtheorem{Example}[Theorem]{Example}
\newtheorem{Definition}[Theorem]{Definition}

\providecommand{\CC}{\mathbb{C}}
\providecommand{\DD}{\mathbb{D}}

\providecommand{\NN}{\mathbb{N}}

\providecommand{\RR}{\mathbb{R}}

\providecommand{\ZZ}{\mathbb{Z}}

\newcommand{\cA}{\mathcal{A}}
\newcommand{\cB}{\mathcal{B}}

\newcommand{\cD}{\mathcal{D}}
\newcommand{\cE}{\mathcal{E}}

\newcommand{\cH}{\mathcal{H}}

\newcommand{\cJ}{\mathcal{J}}
\newcommand{\cK}{\mathcal{K}}
\newcommand{\cL}{\mathcal{L}}

\newcommand{\cO}{\mathcal{O}}

\newcommand{\cT}{\mathcal{T}}

\newcommand{\cW}{\mathcal{W}}

\newcommand{\fA}{\mathfrak{A}}

\newcommand{\fg}{\mathfrak{g}}

\newcommand{\fL}{\mathfrak{L}}

\newcommand{\fr}{\mathfrak{r}}

\newcommand{\sH}{\mathsf{H}}
\newcommand{\sN}{\mathsf{N}}
\newcommand{\sT}{\mathsf{T}}
\newcommand{\sV}{\mathsf{V}}

\providecommand{\abs}[1]{\left \lvert #1 \right \rvert}

\providecommand{\babs}[1]{\bigl \lvert #1 \bigr \rvert}

\providecommand{\norm}[1]{\left \lVert #1 \right \rVert}

\providecommand{\bnorm}[1]{\bigl \lVert #1 \bigr \rVert}

\providecommand{\scpro}[2]{\left \langle #1 , #2 \right \rangle}

\providecommand{\expval}[1]{\left \langle #1 \right \rangle}

\providecommand{\bexpval}[1]{\bigl \langle #1 \bigr \rangle}

\DeclareMathOperator\Renew{Re}
\DeclareMathOperator\Imnew{Im}
\renewcommand{\Re}{\Renew}
\renewcommand{\Im}{\Imnew}

\DeclareMathOperator\vol{vol}
\DeclareMathOperator\Vol{Vol}
\DeclareMathOperator\End{End}

\DeclareMathOperator{\dom}{D}
\DeclareMathOperator\dist{dist}

\DeclareMathOperator{\dive}{div}

\DeclareMathOperator{\sff}{II}
\DeclareMathOperator{\im}{ran}

\providecommand{\dd}{\mathrm{d}}
\providecommand{\dD}{\mathrm{D}}
\providecommand{\e}{\mathrm{e}}
\providecommand{\I}{\mathrm{i}}
\providecommand{\rs}[2]{{\left.#1\right|}_{#2}} 

\newcommand{\sct}{C^\infty}
\newcommand{\pr}{\mathrm{pr}}

\newcommand{\eps}{\varepsilon}
\newcommand{\Ha}{H_\mathrm{a}}

\newcommand{\Heff}{H_\mathrm{eff}}
\newcommand{\id}{\mathbf{1}}

\providecommand{\iec}{i.e.,~}

\SetMathAlphabet{\mathcal}{normal}{OMS}{cmsy}{m}{n}

\DeclareSymbolFont{lettersA}{U}{txmia}{m}{it}
\DeclareMathSymbol{\updelta}{\mathord}{lettersA}{"0E}

\begin{document}

\title{Quantum Waveguides with Magnetic Fields}

\author{Stefan Haag}
\address{Allianz Lebensversicherungs-AG\\ Reinsburgstraße 19, 70178 Stuttgart, Germany} 
\email{stefan.haag@allianz.de}

\author{Jonas Lampart}
\address{CNRS \& Laboratoire Interdisciplinaire Carnot de Bourgogne, UMR 6303 CNRS \& Universit\'e de Bourgogne Franche-Comt\'e, 9 Av. A. Savary, 21078 Dijon CEDEX, France}
\email{jonas.lampart@u-bourgogne.fr}

\author{Stefan Teufel}
\address{Mathematisches Institut, Eberhard Karls Universit\"at T\"ubingen\\ Auf der Morgenstelle 10, 72076 T\"ubingen, Germany}
\email{stefan.teufel@uni-tuebingen.de}

\maketitle

%
%
%
%
%
%

\begin{abstract}
We study generalised quantum waveguides in the presence of moderate and strong external magnetic fields. Applying recent results on the adiabatic limit of the connection Laplacian we show how to construct and compute effective Hamiltonians that allow, in particular, for a detailed spectral analysis of  magnetic waveguide  Hamiltonians. 
We apply our general construction to a number of explicit examples, most of which are not covered by previous results.   
\end{abstract}

\maketitle

\tableofcontents

\section{Introduction}

In this work we study generalised quantum waveguides as introduced in \cite{Haag2015} in the presence of moderate and strong external magnetic fields. In a nutshell, a generalised quantum waveguide is an $\eps$-thin neighbourhood of a submanifold, or a boundary thereof. The corresponding Hamiltonian is the magnetic Laplacian in the interior with Dirichlet boundary conditions, or the Laplacian on the surface. 
The strength of the magnetic field is either held fixed or of the same order as the constraining forces, which increase in the asymptotic limit $\eps\ll 1$. Our results concern the construction of   effective Hamiltonians, defined only on the submanifold, that provide an approximation of  the spectrum of the full waveguide Hamiltonian for small $\eps$. Our construction is based on recent results on the adiabatic limit of the connection Laplacian \cite{Haag2016b} that hold for very general geometries.   Here we analyse the effective Hamiltonians in detail for magnetic quantum waveguides and discuss the relevant terms for different energy scales, geometries, and strengths of the magnetic field.

More specifically,  for us  a quantum waveguide is 
an $\eps$-thin   neighbourhood $\cT^\eps$    of a smoothly embedded $b$-dimensional submanifold $B\xhookrightarrow{c}\RR^{b+f}$ of co-dimension $f$. The size and shape of the cross-section, i.e.\ the intersection of $\cT^\eps$ with a plane normal to the submanifold, may vary along the submanifold.  Common examples are quantum strips ($b=1$ and $f=1$), quantum tubes ($b=1$ and $f=2$) and quantum layers ($b=2$ and $f=1$). Additionally, we consider the case of so-called hollow waveguides, which are modelled on the boundary of a massive waveguide as just described.
\begin{figure}[H]
\centering
\def\svgwidth{0.7\textwidth}
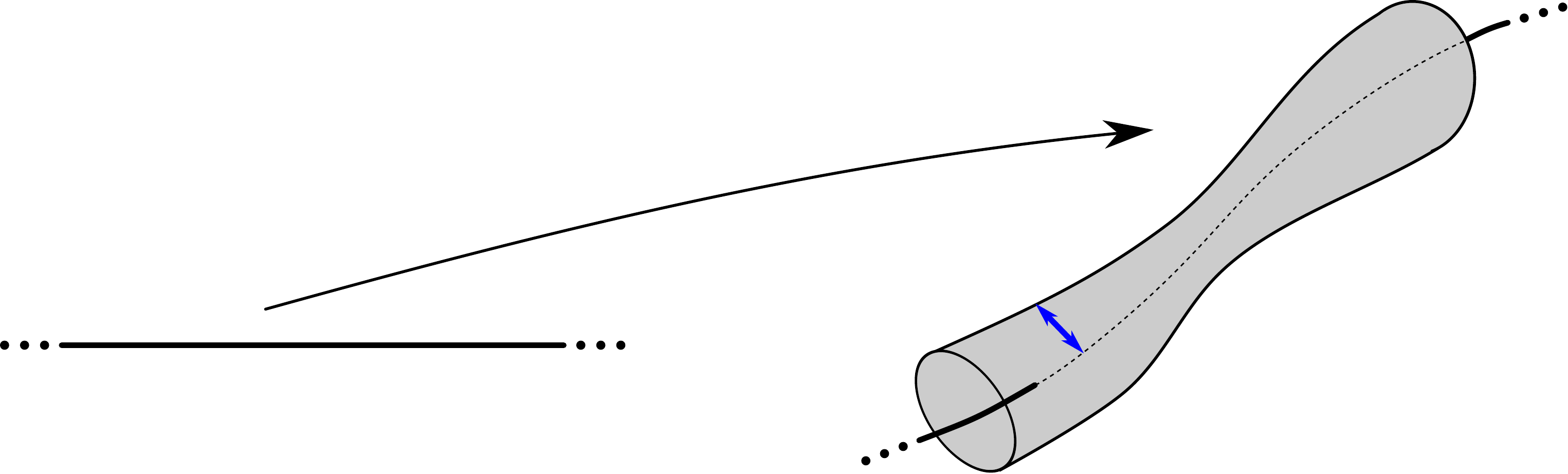
\bigskip
\caption{\small Sketch of a generalised quantum waveguide, modelled as a thin  neighbourhood of a manifold $B$  embedded  into an ambient Euclidean space. The domain of the corresponding massive waveguide  is the interior of the tube, while the domain of the hollow waveguide is its boundary, \iec the surface of the tube.}
\label{fig:QWG}
\end{figure}
In the absence of any magnetic field, the properties of freely moving quantum particles are modelled by the Laplacian
on $L^2(\cT^\eps)$ with Dirichlet boundary conditions for massive wave\-guides, or by the Laplace-Beltrami operator of the induced metric on the boundary of $\cT^\eps$ for hollow waveguides. Since localised and normalised elements of these domains (apart from the constant transversal ground state for hollow waveguides) necessarily have  derivatives of order $\eps^{-1}$ in  the transversal directions, we will introduce the pre-factor $\eps^2$ in front of the Laplacian in order to rescale the energies properly. The effects of an external magnetic field $\cB=\dd \cA$, represented by a magnetic potential $\cA\in\sct(\sT^*\RR^{b+f})$, are incorporated by minimal coupling, replacing the flat connection $\dd$ by the magnetic connection $\nabla^{\eps^{-\sigma}\cA}:=\dd+\I\eps^{-\sigma}\cA$.
The parameter $\sigma\in\{0,1\}$ controls the coupling strength. We refer to $\sigma=0$ as moderate magnetic fields and to $\sigma=1$ as  strong  magnetic fields. Thus, on the macroscopic scale where the waveguide is of width $\eps$, a moderate magnetic field $\cB$  is of order one, while a strong magnetic field $\cB$ is of order $\eps^{-1}$.

There exists a vast literature on specific types of quantum waveguides without magnetic fields, see e.g.\ \cite{Exner2015,Haag2015} and references therein.  There are significantly less results on the magnetic case. For example,    Ekholm  and Kova{\v{r}}{\'\i}k \cite{Ekholm2005} consider quantum strips ($B=\RR$, $f=b=1$) and Grushin \cite{Grushin2008} special tubes ($B=\mathbb{S}^1$, $f=2$), in  both cases with magnetic fields of moderate strength. The more recent work \cite{Krejcirik2014} of Krej\v{c}i\v{r}\'{i}k and Raymond deals with quantum strips and quantum tubes ($B=\RR$, $f=1,2$) of constant cross-section and   moderate as well as strong magnetic fields. Finally, in \cite{Krejcirik2015} the same authors with Tu\v{s}ek treat $\eps$-tubular neighborhoods of hypersurfaces in $\RR^{b+1}$ in the presence of moderate magnetic fields. 

In our approach we are able to treat all these geometries, and actually much more general ones, in a unified framework. For the case of moderate fields, our main result, Theorem~\ref{thm:weak}, states that  for any generalised waveguide  the effective Hamiltonian can be obtained by minimal coupling of the effective non-magnetic Hamiltonian from \cite{Haag2015} to an effective vector potential. For strong magnetic fields the   more complicated effective Hamiltonian is given in Theorem~\ref{thm:strong}.   Note that, in both cases, the quality of the approximation depends on the relevant energy scale of the problem, see Corollary~\ref{cor:main1} and Theorem~\ref{thm:main2}.
In Section~\ref{chap:conv} we explicitly compute and discuss the effective operators for  different  types of massive and hollow quantum tubes ($b=1$, $f=2$) and for moderate and strong magnetic fields. In particular, we recover as a special case the effective Hamiltonian of \cite{Krejcirik2014}, cf.\ Eq.~\eqref{KRHam}. The general effective Hamiltonians of Sections~\ref{chap:unscaled} and \ref{chap:scaled} can be easily evaluated also for all other special cases considered previously in the literature.

As a completely new example,  treated in detail at the end of  Section~\ref{sec:scaled}, let us mention the case of a hollow waveguide $\cT^\eps$ modelled on the surface of a cylindrical tube of varying radius  $\eps \,\ell(x)$ along a smooth curve  $c:\RR\to \RR^3$ in the presence of a strong magnetic field $\eps^{-1}\cB$. 
The resulting effective Hamiltonian is at leading order 
\begin{align*}
H_\mathrm{eff} =  - \partial_x^2 + \tfrac{1}{2}\tfrac{\ell''}{\ell} - \tfrac{1}{4} \left(\tfrac{\ell'}{\ell}\right)^2 +\tfrac{1}{4}\ell^2 \bigl({\cB^\parallel}^2 + 2 \bnorm{\cB^\bot}^2_{\RR^2}\bigr)  \,,
\end{align*}
where $\cB^\parallel$ and $\cB^\perp$ are the projections of $\cB|_c$ onto the tangential resp.\ normal directions of the curve.
The spectrum of  $H_\mathrm{eff} $  is $\eps$-close to the spectrum of the Laplace-Beltrami  operator $-\Delta_{\cT^\eps}^{\eps^{-1}\cA}$ on the surface $\cT^\eps\subset\RR^3$ in the sense that 
\[
\mathrm{dist}\Bigl( \sigma (H_\mathrm{eff}  ) \cap [0,E] , \sigma\bigl(-\Delta_{\cT^\eps}^{\eps^{-1}\cA}\bigr)\cap[0,E]\Bigr) = \cO(\eps)
\]
for all $E\in \RR$ (see Theorem~\ref{thm:main2} and mind the rescaling of energies by a factor~$\eps^2$).


Let us give a slightly more detailed plan of the paper.
In Section~\ref{chap:QWG}, we identify the family $\{\cT^\eps\}_{0<\eps\leq 1}$ (the \enquote{tube}) with an $\eps$-independent manifold $M$ (the \enquote{waveguide}) that has the additional structure of a fibre bundle $M\xrightarrow{\pi_M} B$, with typical fibre (\enquote{typical cross-section}) $F$ given by a compact subset of $\RR^f$. The corresponding diffeomorphism $\Psi_\eps:M\to \cT^\eps$ is described  in 
\cite[Chapter~2]{Haag2015} and  lifts to a unitary operator
$
 L^2(\cT^\eps )\to L^2(M,\vol_{G^\eps}) $. The magnetic waveguide Laplacian $-\eps^2\Delta^{\eps^{-\sigma}\cA}$ on $\cT^\eps\subset\RR^{b+f}$ is thus unitarily equivalent to $-\Delta_{G^\eps}^{\eps^{-\sigma}\cA^\eps}$ on $L^2(M,\vol_{G^\eps})$, with an induced Riemannian metric $G^\eps $ on $M$ and an induced magnetic potential $\cA^\eps \in\sct(\sT^* M)$. The precise form of the two last-mentioned objects will be discussed in Section~\ref{sec:metric} and Section~\ref{sec:magpot}, respectively.  In Section~\ref{sec:operator} we finally show that, after changing the volume measure,  $-\eps^2\Delta^{\eps^{-\sigma}\cA}$ is  unitarily equivalent to the operator 
\[
H^\sigma = - \eps^2\fL \left(G^\eps_\mathrm{hor}, \eps^{-\sigma}\cA^\eps_\mathrm{hor}\right) + H^{\sV,\sigma} + \eps^2 V_\mathrm{bend}
\]
on an $\eps$-independent Hilbert space of functions on $M$. It  is the sum of a horizontal operator $- \eps^2\fL (G^\eps_\mathrm{hor}, \eps^{-\sigma}\cA^\eps_\mathrm{hor})$ that acts \enquote{tangentially} to $B$, a vertical operator $H^{\sV,\sigma}$ (depending on the vertical contribution $\cA_\sV^\eps$ of the induced magnetic potential) that acts transversally to $B$, and the so-called bending potential $V_\mathrm{bend}$.
The operator $H^\sigma$ is   of the general form 
 to which the results of \cite{Haag2016b} can be applied. The    statements resulting from \cite{Haag2016b} concerning  the  effective Hamiltonian and its approximation properties are   discussed in Section~\ref{chap:results}. 
There we also introduce  the so-called adiabatic Hamiltonian, which is in some sense the simplest and most explicit  ansatz for an effective Hamiltonian.

In Section~\ref{chap:unscaled} and Section~\ref{chap:scaled}, we compute the relevant terms in an asymp\-totic expansion of the adiabatic Hamiltonian for moderate and strong magnetic fields, respectively. The form of these expansions and the approximation errors depend, among other things, on the energy scale under consideration. We consider energies at distance of order $\eps^\alpha$ above the infimum $\Lambda_0$ of the spectrum of the non-magnetic vertical operator $H^\sV_0=\rs{H^{\sV,\sigma}}{\cA_\sV^\eps=0}$ for $\alpha\in[0,2]$. So $\alpha=0$ corresponds to energies of order one in units of the level spacing of transversal eigenvalues and $\alpha=2$ to energies $\eps^2$-close to the infimum of the spectrum. 
Roughly speaking, the  spectrum of the adiabatic Hamiltonian turns out to be $\eps^{2+\alpha}$-close to the spectrum of the magnetic waveguide Laplacian. Therefore, in the asymptotic expansion of the adiabatic Hamiltonian terms of order $\eps^{2+\alpha}$ and smaller can be discarded. Whether the interval $[\Lambda_0,\Lambda_0 + C\eps^\alpha]$ for $\alpha>0$ contains any spectrum of the magnetic waveguide Hamiltonian at all depends on the geometry of the waveguide.  

The main result of Section~\ref{chap:unscaled}, Theorem~\ref{thm:weak}, states that for moderate magnetic fields the adiabatic Hamiltonian can be obtained by minimal coupling of the non-magnetic adiabatic Hamiltonian as computed in \cite{Haag2015} to an effective magnetic potential, for all values of $\alpha\in[0,2]$.
For strong magnetic fields the adiabatic Hamiltonian is not merely altered by minimal coupling to an effective magnetic potential, but the effects of the field are more subtle, see Theorem~\ref{thm:strong} and Corollary~\ref{cor:strong}.

In Section~6 we finally consider the special case of quantum tubes, \iec the case $b=1$ and $f=2$, where the expansions of the adiabatic Hamiltonian become even more explicit.

\section{Generalised Quantum Waveguides} \label{chap:QWG}

Let $c:B\to \RR^{b+f}$ be the smooth embedding of a smooth, complete $b$-dimensional manifold $B$ into the Euclidean space $\RR^{b+f}$ with the standard metric $\updelta$. 
%
We have an orthogonal decomposition into vectors tangent and normal to $c(B)$, that is
\begin{equation}
\label{eq:TBNB}
c^* \sT \RR^{b+f} = \sT B \oplus \sN B.
\end{equation}
The Riemannian metric~$\updelta$ and its Levi-Civita connection $\nabla^{\updelta}$ give rise to
\begin{itemize}
\item a Riemannian metric $g_B=c^*\updelta$ on $B$ and its associated Levi-Civita connection~$\nabla^{g_B}$,
\item a bundle metric $\updelta^{\sN B}=\updelta \vert_{\sN B}$ and a metric connection $\nabla^{\sN B}$ on $\sN B$, where the metric is obtained by restriction of the Euclidean metric and the connection is the projection of $\nabla^\updelta$ to the subbundle $\sN B$.
\end{itemize}
Suppose that there exists a tubular neighbourhood $\cT^r\subset\RR^{b+f}$ of $c(B)$ with globally fixed radius    $r>0$ that does not self-intersect.
This is equivalent to the requirement that the map
\[
\Phi:\sN B\to \RR^{b+f},\quad \sN_x B\ni \nu\mapsto c(x) + \nu
\]
restricted to
\[
\sN B^r := \bigl\{\nu\in\sN B\text{ such that }\norm{\nu}_{\updelta^{ \sN B}}< r\bigr\}
\]
be a diffeomorphism onto $\cT^r$. If this holds, the entire analysis can be carried out on (a subset of) the normal bundle. This suggests that we view the initial tube as an $\eps$-independent, fibrewise embedded submanifold $\varpi:M\to\sN B^r$, which is then mapped into $\RR^{b+f}$ by means of a rescaled diffeomorphism
\[
\Phi_\eps:\nu \mapsto c(\pi_{\sN B}(\nu))+\eps{\nu},\quad 0<\eps \leq 1.
\]
The composition
\begin{equation}
\label{eq:Psieps}
\Psi_\eps := \Phi_\eps\circ \varpi,\quad 0 < \eps\leq 1
\end{equation}
finally yields the desired change of perspective from the family of $\eps$-thin tubes $\{\cT^\eps\}_{0<\eps\leq 1}$ to the $\eps$-independent waveguide $M$ with $\eps$-dependent family of rescaled pullback metrics $\{G^\eps:=\Psi_\eps^*(\eps^{-2}\updelta)\}_{0<\eps\leq 1}$.

\begin{Definition} \label{def:QWG}
Let $B\xhookrightarrow{c} \RR^{b+f}$ be a smooth, complete, embedded $b$-dimen\-sional submanifold with tubular neighbourhood $\cT^r\subset\RR^{b+f}$ of fixed radius $r>0$, and $F$ be a compact manifold of dimension $\dim(F)\leq f$ with smooth boundary. Let $M\xhookrightarrow{\varpi}\sN B^r$ be a connected submanifold that has the structure of a fibre bundle $\pi_M:M\to B$ with typical fibre $F$ such that the diagram
\[
\begin{diagram}
M & \rTo^{\varpi} & \sN B^r \\
\dTo^{\pi_M} & & \dTo_{\pi_{\sN B}} \\
B & \rTo_{\id_B} & B
\end{diagram}
\]
commutes. We then call $M$ a \emph{generalised quantum waveguide}.
\end{Definition}

\begin{figure}[H]
\centering
\def\svgwidth{0.9\textwidth}
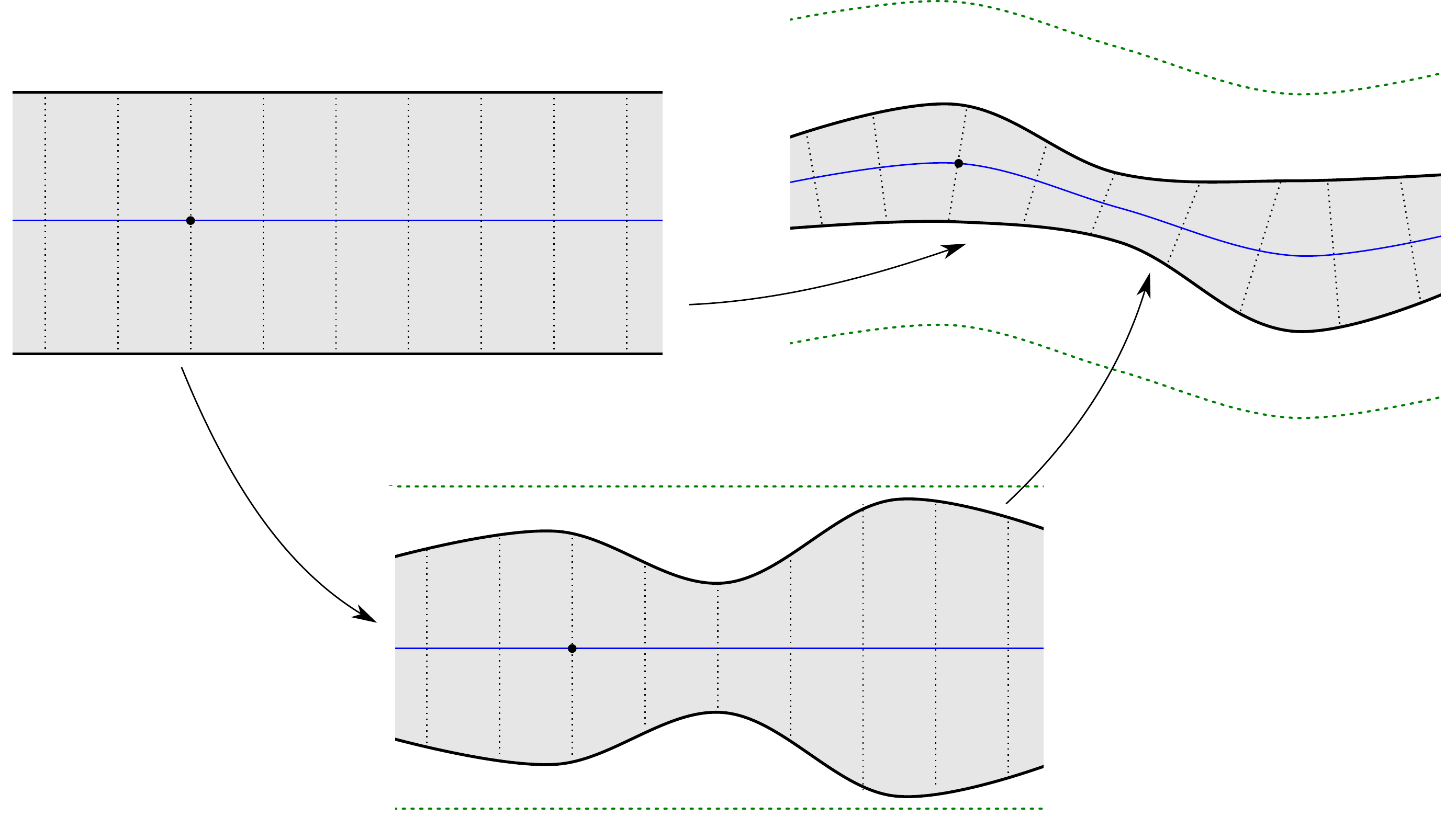
\caption{\small The two embeddings $B\xhookrightarrow{c} \RR^{b+f}$ and $M\xhookrightarrow{\varpi} \sN B^r\subset\sN B$ allow for the identification of the fixed waveguide $M\xrightarrow{\pi_M} B$ with the family of $\eps$-thin tubes $\cT^\eps\subset\RR^{b+f}$.}
\label{fig:embedding}
\end{figure}
\noindent We will refer to the fibres $M_x:=\pi_M^{-1}(x)$, $x\in B$, of the bundle $M\xrightarrow{\pi_M} B$ as the cross-sections of the waveguide, being fibrewise submanifolds of the respective normal spaces $\sN_x B$. The most interesting examples of generalised quantum waveguides are given by (see \cite[Definition 2.2]{Haag2015})
\begin{itemize}
\item \emph{massive quantum waveguides, $\dim(F)=f$:} \\
the typical fibre $F$ of $M$ is the closure of an open, bounded and connected subset of $\RR^f$ with smooth boundary,
\item \emph{hollow quantum waveguides, $\dim(F)=f-1$:} \\
$f\geq 2$ and the waveguide $M$ is the (fibrewise) boundary of a massive waveguide.
\end{itemize}

We impose the following uniformity conditions on the geometry of the generalized quantum waveguide and the magnetic potential:

\begin{Definition} \label{def:QWGbound}
Let $M$ be a generalised quantum waveguide as in Definition~\ref{def:QWG} and $\cA\in\sct(\sT^*\RR^{b+f})$ be a magnetic potential.
\begin{enumerate}
\item We call $M$ a quantum waveguide of bounded geometry (cf.~\cite[Definition~3.1]{Haag2015}) if
\begin{itemize}
\item $(B,g_B)$ is a manifold of bounded geometry (see \cite[Appendix~A1.1]{Shubin1992}),
\item $\pi_M:(M,g)\to(B,g_B)$ is a uniformly locally trivial fibre bundle (see \cite[Definiton~3.2]{Lampart2016}),
\item the embeddings $(B,g_B)\xhookrightarrow{c}(\RR^{b+f},\updelta)$ and $(M,g)\xhookrightarrow{\varpi} (\sN B^r,\Phi^*\updelta)$ are bounded with all their derivatives.
\end{itemize}
Here, 
$g:=g^{\eps=1}$ is the unscaled version of the submersion metric $g^\eps$~\eqref{eq:geps} on~$M$  and $\Phi=\Phi_{\eps=1}$.
\item We call the pair $(M,\cA)$ a magnetic quantum waveguide of bounded geometry if $M$ is a quantum waveguide of bounded geometry with fixed tubular radius $r>0$ and the restricted one-form $\rs{\cA}{\cT^r}$ is bounded with all its derivatives.
\end{enumerate}
\end{Definition}
These conditions on $M$ are trivially satisfied for compact manifolds and in many concrete examples, such as \enquote{asymptotically straight} or periodic waveguides.

We will now discuss the geometry of such quantum waveguides in some more detail, in order to get an explicit expression for the operator $H^\sigma$ on $L^2(M,\vol_g)$, obtained from the Laplacian on $\cT^\eps$ with magnetic potential $\eps^{-\sigma}\cA$ via $\Psi_\eps$.

\subsection{The Pullback of the Riemannian Metric} \label{sec:metric}

Let $\sV M := \ker(\sT\pi_M)$ be the vertical subbundle of $\sT M$. Its elements are vectors that are tangent to the fibres of $M$. 
For a massive waveguide, these vectors are tangent to the fibres of $\sN B$, so 
they point along the normals to $B$. 
They are thus naturally identified with $\sN B$ by \enquote{forgetting the 
base point}, mapping $\sT_\nu (\sN_x B) \to \sN_x B$ in the obvious way. Thus, 
in the massive case, $\sV M \cong \sN B$.   
This identification is formalised by introducing the map $\cK_{\sN B}:\sT(\sN B)\to\sN B$, which projects to vertical vectors, orthogonally with respect to $\updelta^{\sN B}$, and then uses the identification above.  

The orthogonal complement of the vertical subbundle $\sV M$ (with respect to $G=G^{\eps=1}$) is called the horizontal subbundle $\sH M \cong \pi^*_M\sT B$. It turns out that the choice of horizontal bundle does not depend on~$\eps$, that is, the decomposition
\[
\sT M = \sH M \oplus \sV M,
\]
is $G^\eps$-orthogonal for every $\eps>0$ (see \cite[Equation~(4.5)]{Haag2015} for the massive case and \cite[Lemma~5.2]{Haag2015} for the hollow case).
The pullback metric $G^\eps$ can then be written as the sum of a horizontal and 
a vertical part. More precisely, it is of the form
\begin{equation}
\label{eq:Geps}
G^\eps = \eps^{-2} \underbrace{\bigl(\pi_M^* g_B + \eps h^\eps\bigr)}_{=:G^\eps_\mathrm{hor}} + g_\sV.
\end{equation}
Here, $g_\sV:=\rs{G^\eps}{\sV M}$ is the $\eps$-independent restriction of 
$G^\eps$ to the vertical subbundle. The restriction $g_{M_x}:=\rs{g_\sV}{M_x}$ 
defines a Riemannian metric on each fibre $M_x$ of $M$. Moreover, $h^\eps\in 
\sct(\sT^*M^{\otimes 2})$ is a symmetric (but not necessarily non-degenerate) 
tensor that satisfies $h^\eps(V,\cdot)=0$ for all $V\in\sct(\sV M)$ and encodes corrections to the geometry due to the finite diameter  $\eps$ of the tube $\cT^\eps$. 
Thus, if we 
define for any vector field $X\in\sct(\sT B)$ its horizontal lift  
$X^\sH\in\sct(\sH M)$ as the unique horizontal vector field satisfying $\sT\pi_M (X^\sH) = X$, we have
\begin{align*}
G^\eps(X^\sH,Y^\sH) &=\eps^{-2}\bigl(g_B(X,Y)+\eps h^\eps(X^\sH,X^\sH)\bigr) , \\
G^\eps(X^\sH,V) &= 0 \, , \\
G^\eps(V,W) &= g_\sV(V,W)
\end{align*}
for all $X,Y\in\sct(\sT B)$ and $V,W\in \sct(\sV M)$. The explicit form~\eqref{eq:Geps} shows that the pullback metric $G^\eps$ may be considered as a small perturbation of the metric
\begin{equation}
\label{eq:geps}
g^\eps:= \eps^{-2} \pi_M^* g_B + g_\sV.
\end{equation}
We will refer to this metric as the (rescaled) submersion metric, as it makes 
$\sT \pi_M$ an isometry from $\sH M$ to $\sT B$ (up to a factor $\eps^{-1}$), 
so $g:=g^{\eps=1}$ is a Riemannian submersion.
We finally mention that the horizontal deviation $h^\eps$ can be computed explicitly from the quantities associated with the embeddings $B\hookrightarrow \RR^{b+f}$ and $M \hookrightarrow \sN B$. We have
\begin{align*}
h^\eps(X^\sH,Y^\sH)  =& - 2 g_B\bigl(\varpi^*\sff(\cdot)X,Y\bigr)  + \eps \Bigl[g_B\bigl(\varpi^*\cW(\cdot)X,\varpi^*\cW(\cdot)Y\bigr)  \\
& + \underbrace{ \updelta^{\sN B}\bigl(\cK_{\sN B}\circ\gimel(X),\cK_{\sN B}\circ\gimel(Y)\bigr)}_{\text{appears only in the hollow case}}\Bigr]
\end{align*}
for all $X,Y\in\sct(\sT B)$, where
\begin{itemize}
\item $\cW\in\sct(\sN^* B\otimes\End(\sT B))$ is the Weingarten map of the embedding $B\hookrightarrow \RR^{b+f}$ and $\sff\in\sct(\sN ^*B\otimes \sT^*B^{\otimes 2})$ is the associated second fundamental form $\sff(\cdot)(X,Y):=g_B(\cW(\cdot)X,Y)$,
\item and $\gimel(X)$ 
is a vertical vector that
 measures the difference between the horizontal lifts with respect to the hollow waveguide and its underlying massive waveguide (cf.~\cite[Section 5.1]{Haag2015} and Eq.~\eqref{eq:gimel} below). 
\end{itemize}

\subsection{The Pullback of the Magnetic Potential}  \label{sec:magpot}

In this section we give explicit expressions for the pulled-back magnetic potential $\Psi_\eps^*\cA$, in a specific gauge. The role of this gauge is to make the effect of the $\eps$-scaling of the fibres on the magnetic potential more explicit. The idea is the following: On $\cT^\eps$ perform a Taylor expansion of $\cA$ with respect to the normal directions
\begin{equation*}
 \cA_{c(x) + \eps\nu}=\cA_{c(x)} + \eps\nu \widetilde \cA.
\end{equation*}
By a gauge transformation, we can make the one-form $\cA_{c(B)}$ vanish on normal vectors, turning it into a one-form on the tangent space to $B$. 
When $\cA$ is written in this way, one sees easily that the leading term is intrinsic to $B$. It can thus naturally be integrated into the effective operator and also commutes with any differential operators acting in the normal directions.

\begin{Lemma}\label{lem:gauge}
 Let $r>0$ and $\cA \in C^\infty(\sT^* \cT^r)$ be a one-form with globally bounded derivatives on $\cT^r$. Define the function
\[
\Omega_\eps:\sN B^r\to\RR,\quad \sN_x B^r\ni\nu\mapsto \cA_{c(x)}\bigl(\eps {\nu}\bigr),
\]
and let $\cA_B:=c^*\cA$. Then 
\begin{align}
 \cA^\eps:=&\Psi_\eps^*\cA - \dd \varpi^*\Omega_\eps= \varpi^*\left(\Phi_\eps^*\cA-\dd \Omega_\eps\right)\label{eq:Aeps def}\\
 =&\pi_M^*\cA_B + \eps \cA_\sH^\eps + \eps^2 \cA_\sV^\eps,\notag
\end{align}
where $\cA_\sH^\eps$ vanishes on the vertical vectors $\sV M$,  $\cA_\sV^\eps$ vanishes on $\sH M$ and both have globally bounded derivatives with respect to the metric $g=g^{\eps=1}$, uniformly in $\eps$.
\end{Lemma}
\begin{proof}
To evaluate $\dd \Omega_\eps$ on vertical vectors, take $v\in \sT_\nu (\sN_x B^r) = \sV_\nu \sN B^r$ and let $\gamma(t)=\nu + t \cK_{\sN B}v$. Then
\begin{equation*}
\dd\Omega_\eps(v)=\eps \rs{\frac{\dd}{\dd t}}{t=0}\cA_{c(x)}(\gamma(t))=\cA_{c(x)}\left( \eps \cK_{\sN B}v \right)
\end{equation*}
By the same reasoning, $\sT\Phi_\eps(v)=\eps \cK_{\sN B}v$, so we have
\begin{align*}
 \left(\Phi_\eps^*\cA-\dd \Omega_\eps\right)(v)&=\eps\left(\cA_{\Phi_\eps(\nu)} - \cA_{c(x)}\right)( \cK_{\sN B}v)\\
 &=\eps\left(\cA_{c(x)+\eps\nu} - \cA_{c(x)}\right)( \cK_{\sN B}v).
\end{align*}
This is clearly of order $\eps^2$, and $\varpi^*$ is a partial isometry. 
To calculate the horizontal component, take $w \in \sT_x B$ and denote by $w^\#$ its horizontal lift to $\sT_\nu N B$ for $\nu \in \sN_x B$. Note that $w^\#=\sT\varpi (w^\sH)$ for massive waveguides, but not in general. We have $\sT \Phi_\eps(w^\#)=w-\eps \cW(\nu)w$, see e.g.~\cite[Section 4.1]{Haag2015}. This gives
\begin{equation*}
 \Phi^*_\eps\cA(w^\#)= (c^*\cA)(w) + \bigl(\cA_{c(x)+\eps\nu} - \cA_{c(x)}\bigr)(w)- \eps \cA_{c(x)+\eps\nu}\bigl(\cW(\nu)w\bigr).
\end{equation*}
The first term is $\cA_B$, while the remaining ones are clearly of order $\eps$. As $\dd \Omega_\eps = \eps \dd \Omega_1$ this proves the claim.
\end{proof}
The explicit expressions for $\cA_\sH^\eps$ and $\cA_\sV^\eps$ are easily obtained from this proof. To leading order we have, for a vertical field $V \in C^\infty(\sV M)$,
\begin{align*}
\rs{\cA_\sV^\eps(V)}{\varpi^{-1}(\nu)}=(\dD_{\nu}\cA)\bigl(\cK_{\sN B}\circ \sT\varpi (V)\bigr) + \cO(\eps),
\end{align*}
where $\dD_{\nu}\cA$ is the Lie derivative of $\cA$ in the direction of the (constant) vector field $\nu$ on $\RR^{b+f}$.
For the horizontal part, we let $X\in C^\infty(\sT B)$
and obtain (for $\nu \in \sN_x B$)
\begin{equation*}
\rs{\cA_\sH^\eps(X^\sH)}{\varpi^{-1}(\nu)}=(\dD_{\nu}\cA)(X) - \cA_{c(x)}\bigl(\cW(\nu)X\bigr) + \cO(\eps).
\end{equation*}
Here, we have used that the vertical vector
\begin{equation}
\label{eq:gimel}
\gimel(X):=\eps^{-1}\bigl(\sT\varpi(X^\sH)  - X^\#\bigr)
\end{equation}
has length of order one, see~\cite[Section 5.1]{Haag2015}.

\subsection{The Induced Operator} \label{sec:operator}

The operator we want to study is the Laplacian with magnetic potential $\eps^{-\sigma}\cA$ of the waveguide. For massive waveguides this means the Laplacian on the tube $\cT^\eps$, with Dirichlet boundary boundary conditions. For hollow waveguides this is the Laplace-Beltrami operator on the submanifold $\Psi_\eps(M)\subset \RR^{b+f}$, with the induced metric.
In both cases, the operator is, up to a global factor $\eps^2$, equivalent to the Laplacian on $L^2(M,\vol_{G^\eps})$ with magnetic potential $\eps^{-\sigma}A^\eps$ (which is gauge-equivalent to $\eps^{-\sigma}\Psi^*_\eps \cA$, see~\eqref{eq:Aeps def}), that is, the operator whose quadratic form is given by
\begin{equation*}
\langle \psi, -\Delta_{G^\eps}^{\eps^{-\sigma}\cA^\eps} \psi \rangle_{L^2(M, \vol_{G^\eps})}
= \int\limits_M G^\eps \left(\overline{(\dd+ \I \eps^{-\sigma}\cA^\eps) \psi},(\dd+ \I \eps^{-\sigma}\cA^\eps) \psi\right)  \vol_{G^\eps}
\end{equation*}
for $\psi\in W^{1,2}_0(M, G^\eps)$ (which equals $W^{1,2}(M,G^\eps)$ in the hollow case). 

It is convenient to work on the $\eps$-independent Hilbert space $\cH:=L^2(M,\vol_g)$, where $g=g^{\eps=1}$ is the unscaled submersion metric. This will also make many effects of the scaling and the geometry appear more explicitly in the operator. The correspondence is given by the unitary map
\begin{equation*}
 {U}_{\rho_\eps} : L^2(M,\vol_{G^\eps}) \to \cH, \quad \psi \mapsto \bigl(\eps^{-b}\rho_\eps\bigr)^{1/2}\psi,
\end{equation*}
where 
\begin{equation}
\label{eq:rhoeps}
\rho_\eps :=  \frac{\vol_{G^\eps}}{\eps^{-b}\vol_g} = \frac{\vol_{G^\eps}}{\vol_{g^\eps}} = \frac{\vol_{G_\mathrm{hor}^\eps}}{\vol_{\pi_M^* g_B}} \stackrel{\eqref{eq:Geps}}{=}  1 - \eps \vartheta_\eps,
\end{equation}
with $\vartheta_\eps=\cO(1)$ in $C^\infty_\mathrm{b}(M)$, is the relative density of the two measures. The factor $\eps^{-b}$ is due to the scaling of the total volume and does not play any role in the operator.

The transformed operator can be split into an horizontal differential operator, a vertical operator (with Dirichlet conditions), and an $\eps$-dependent potential.
\begin{align}
\label{eq:H}
H^\sigma:=& \; {U}_{\rho_\eps} \bigl(-\Delta_{G^\eps}^{\eps^{-\sigma}\cA^\eps}\bigr){U}_{\rho_\eps}^\dagger \\
=&\; \eps^2\fL\left(G^\eps_\mathrm{hor},\eps^{-\sigma}\cA^\eps_\mathrm{hor}\right) 
+\fL\left(g_\sV,\eps^{2-\sigma}\cA^\eps_\sV\right) 
+ V_{\rho_\eps} \notag.
\end{align}
The objects appearing in \eqref{eq:H} are defined as follows:
\begin{itemize}
\item For a symmetric two-tensor $\fg$ and a one-form $\fA$ on $M$ the differential operator $\fL(\fg,\fA) $ is defined by   its quadratic form 
\begin{equation}\label{eq:Deltahor}
\scpro{\Psi}{\fL(\fg,\fA) \Psi}_\cH
=\int_M \fg
\bigl(\overline{\nabla^\fA\Psi},
\nabla^\fA \Psi\bigr)  \vol_{g},\quad \nabla^\fA:= \dd+\I\fA
\end{equation}
defined on $W^{1,2}_0(M,g)$.
\item 
  The horizontal magnetic potential is given by 
(cf. Equation~\eqref{eq:Aeps def}) 
\begin{equation*}
 \cA_\mathrm{hor}^\eps:= \pi_M^*\cA_B + \eps \cA^\eps_\sH.
\end{equation*}
The factor $\eps^2$ in front of the horizontal operator $\fL\left(G^\eps_\mathrm{hor},\eps^{-\sigma}\cA^\eps_\mathrm{hor}\right)$  is due to the fact that, on the co-vectors $\sT^* M$, $G^\eps=\eps^2 G^\eps_{\rm hor} + g_\sV$.
\item The restriction of the vertical operator to the fibre of $M_x$ over $x$ acts as (minus) the Laplacian associated to the metric $\rs{g_\sV}{M_x}=g_{M_x}$ and the magnetic potential $\eps^{2-\sigma} \cA_\sV^\eps$, that is:
\[
\rs{\fL\bigl(g_\sV,\eps^{2-\sigma}\cA_\sV^\eps\bigr)}{M_x}
= -\Delta_{g_{M_x}}^{\eps^{2-\sigma}\cA_\sV^\eps},
\]
note our sign convention. It is self-adjoint on the Dirichlet domain 
\begin{equation*}
 \cD(M_x):=W^{2,2}(M_x, g_{M_x})\cap W^1_0(M_x, g_{M_x}).
\end{equation*}

\item The geometric potential $V_{\rho_\eps}$ arises from the change of the volume measure and captures effects of the extrinsic curvature of the embedding $B\hookrightarrow \RR^{b+f}$. It takes the form 
\begin{align*}
V_{\rho_\eps} &= \tfrac{1}{2}\Delta_{G^\eps} \ln\rho_\eps - \tfrac{1}{4} G^\eps(\dd\ln\rho_\eps,\dd\ln\rho_\eps) \\
& = \tfrac{1}{2}(-\eps^2 \fL(\pi^*_M g_B, 0) - \fL(g_\sV, 0))\ln\rho_\eps + \tfrac{1}{4} g_\sV(\dd\ln\rho_\eps,\dd\ln\rho_\eps) + \cO(\eps^4),
\end{align*}
which is the same as for $\cA=0$, see~\cite[Equation~(2.5)]{Haag2015}. 
This potential is of order $\eps^2$ for massive waveguides and of order~$\eps$ for hollow waveguides.

\end{itemize}



The operator $H^\sigma$ is self-adjoint on the Dirichlet domain 
\[
\dom(H^\sigma)=W^{2,2}(M,g)\cap W^{1,2}_0(M,g)\,.
\] 
%
\begin{Remark} \label{rem:vertical}
Due to the fact that the magnetic potential in the vertical operator is of order $\eps^{2-\sigma}$, we will be able to treat this perturbatively, even for $\sigma=1$. 
We have
\begin{equation*}
 \Delta_{g_{M_x}}^{\eps^{2-\sigma}\cA_\sV^\eps}= \Delta_{g_{M_x}} + \cO(\eps^{2-\sigma}),
\end{equation*}
and the error will be uniform in $x$ under our hypothesis.
Physically speaking, this effect occurs because the shrinking of the initial tube $\cT^\eps$ in the transversal directions leaves the magnitude of the magnetic field unchanged, while the influence of the vertical differential operator increases.
The perturbation by $\cA_\sV^\eps$ will play a role in the effective operator, if we aim for high precision when $\sigma=0$, and in general when $\sigma=1$.
\end{Remark}
\begin{Remark}
We will exclusively work with $L^2$-Sobolev spaces, so we will drop the corresponding superscript from now on and write $W^{k,2}=W^k$.
When these spaces are associated with the Riemannian manifolds $(M,g)$, $(B,g_B)$ and $(M_x,g_{M_x})$, for $x\in B$, we will omit the dependence on the metric in the notation.
In our setting, these manifolds (with boundary) are of bounded geometry in the sense of \cite{Schick1996}, and the Sobolev norms are defined by summing local norms~\cite[Defintion~3.23]{Schick1996}:
\[
\norm{\psi}^2_{W^k}:= \sum_{\nu\in\ZZ} \norm{(\chi_\nu\psi)\circ\kappa_\nu^{-1}}^2_{W^k(\kappa_\nu(U_\nu))},
\] 
where $\{U_\nu\}_{\nu\in\ZZ}$ is an atlas of normal coordinate charts $\kappa_\nu:U_\nu\to\RR^n$, adapted to the metric, with subordinate partition of unity $\chi_\nu$. 
\end{Remark}

\section{Adiabatic perturbation theory} \label{chap:results}
%
%
%

In this section we discuss the approximation of the spectrum of $H^\sigma$ by that of the adiabatic operator $\Ha^\sigma$, which is an operator on $L^2(B)$. Our results rely on the general construction of (super-)adiabatic approximations  in~\cite{Haag2016b}, and we will refer to that paper for most of the proofs. The general strategy of adiabatic perturbation theory (for pseudo-differential operators) is reviewed, e.g., in \cite{Martinez2007}.

Before going into the technical details, let us briefly sketch the basic ideas of our approach. 
We rewrite  the Schr\"{o}dinger operator~\eqref{eq:H} in the form 
\begin{align*}
H^\sigma =& \; \eps^2\fL\left(G^\eps_\mathrm{hor},\eps^{-\sigma}\cA^\eps_\mathrm{hor}\right) 
+\fL\left(g_\sV,\eps^{2-\sigma}\cA^\eps_\sV\right) 
+ V_{\rho_\eps} \\
=& \;  \eps^2 \fL\left(G_\mathrm{hor}^\eps,\eps^{-\sigma}\cA_\mathrm{hor}^\eps\right) + \eps^2 V_\mathrm{bend} + H^{\sV}_\cA + \cO(\eps^4)
\end{align*}
where $H^{\sV}_\cA$ is a  fibrewise vertical operator
 that is a perturbation of $\fL(g_\sV, 0)$ and $V_\mathrm{bend}$ is a potential, see Equations \eqref{eq:HAV} and \eqref{eq:Vbend}. Let, for $x\in B$, 
\begin{equation*}
 \lambda_\cA(x):=\min\sigma(H^{\sV}_\cA(x))
\end{equation*}
be the ground-state energy of its restriction to $L^2(M_x)$. Let $P_\cA(x)$ be the projection to the associated eigenspace. Since, under our assumptions, the ground-state is a simple eigenvalue, the ranges of $P_\cA(x)$ define a line bundle over $B$. This line bundle is in fact trivial, because the range of $P_\cA(x)$ is close to the ground-state eigenspace of $\fL(g_\sV, 0)\vert_{M_x}$.
The projection $P_\cA$ defines an operator on $\cH$ by setting, for $x\in B$ and $y\in M_x$, $(P_\cA\psi)(y)=(P_\cA(x)\psi)(y)$. The adiabatic operator is then given by
\begin{equation}
\label{eq:Ha}
\Ha^\sigma := P_\cA H^\sigma P_\cA = \eps^2 P_\cA \bigl( \fL\left(G_\mathrm{hor}^\eps,\eps^{-\sigma}\cA_\mathrm{hor}^\eps\right) + V_\mathrm{bend}\bigr) P_\cA + \lambda_\cA P_\cA
\end{equation}
By the triviality of the line bundle associated with $P_\cA$, we may consider this as an operator on $L^2(B)$. 
Consider the difference 
\begin{equation}
(H^\sigma -\Ha^\sigma) P_\cA 
= \eps^2 \bigl[ \fL\left(G_\mathrm{hor}^\eps,\eps^{-\sigma}\cA_\mathrm{hor}^\eps\right) + V_\mathrm{bend},P_\cA\bigr]P_\cA.
\label{eq:comm0}
\end{equation}
If we can show that this is small, the adiabatic operator provides an approximation of $H^\sigma$ on the range of $P_\cA$.

For $\sigma=0$ one sees rather easily that
\begin{equation*} 
[\eps (\dd + \I \cA^\eps_{\mathrm{hor}}), P_\cA]= [\eps (\dd + \I \pi^*_M\cA_B + \eps\cA^\eps_{\sH}), P_\cA]= \cO(\eps),
\end{equation*}
as $\I\eps  \cA^\eps_{\mathrm{hor}}$ is of order $\eps$ by itself and $[\dd, P_\cA]$ can be bounded uniformly in $\eps$, as the eigenfunctions of $H^\sV_\cA(x)$ vary smoothly with $x$. To obtain the same result for $\sigma=1$ it is
important to recall that $\eps^{-1}\cA^\eps_{\mathrm{hor}}=\eps^{-1}\pi^*_M\cA_B + \cA^\eps_{\sH}$, and then note that $\pi^*_M\cA_B$ commutes with the projection $P_\cA$, because it does not depend on the fibre-coordinates.
%
%
This implies that~\eqref{eq:comm0} is of order $\eps$, since $\fL\left(G_\mathrm{hor}^\eps,\eps^{-\sigma}\cA_\mathrm{hor}^\eps\right)$ is a quadratic expression in $\dd + \I \eps^{-\sigma}\cA^\eps_{\mathrm{hor}}$.
So in both cases~\eqref{eq:comm0} is of order $\eps$, as an operator from $\dom(H^\sigma)$ to $\cH$ (it is not of order $\eps^2$, as the graph-norm of $H^\sigma$ is also $\eps$-dependent).
Consequently, the adiabatic operator provides an approximation of $H^\sigma$, with errors of order $\eps$. 

However, since the magnetic and geometric terms in the operator $H^\sigma$ are of order $\eps^{1-\sigma}$ and $\eps^2$, respectively, it is desirable to construct a slightly tilted super-adiabatic projection $P_\cA^\eps = P_\cA+\cO(\eps)$, satisfying $[H^\sigma,P_\cA^\eps]P_\cA^\eps=\cO(\eps^N)$, in a suitable sense, and an intertwining unitary operator $U^\sigma_\eps$ on $\cH$ with $P_\cA^\eps U_\eps = U_\eps P_\cA$, such that the effective operator
\[
\Heff^\sigma := \bigl(U_\eps^\sigma\bigr)^\dagger P_\cA^\eps H^\sigma P_\cA^\eps U^\sigma_\eps = P_\cA \bigl(U_\eps^\sigma\bigr)^\dagger H^\sigma U^\sigma_\eps P_\cA
\]
provides a better approximation of $H^\sigma$ than $\Ha^\sigma$ does. The details of the construction of $P_\cA^\eps$ can be found in~\cite{Haag2016b}.

In the construction just described the choice of a vertical operator $H^{\sV}_\cA$ is not unique.
We make the choice of including all vertical differential operators in $H^{\sV}_\cA$ and, for the case of hollow wave guides, include also parts of the potential $V_{\rho_\eps}$, 
\begin{equation}\label{eq:HAV}
H^{\sV}_\cA := \begin{cases} 
\fL(g_\sV,\eps^{2-\sigma}\cA_\sV^\eps),&\quad \text{mas.} \\
\fL(g_\sV,\eps^{2-\sigma}\cA_\sV^\eps)- \tfrac{1}{2}\fL(g_\sV, 0)\ln\rho_\eps + \tfrac{1}{4}g_\sV(\dd\ln\rho_\eps,\dd\ln\rho_\eps) &\quad\text{hol.} \end{cases}\,,
\end{equation}
and  
\begin{equation}
\label{eq:Vbend}
\eps^2 V_\mathrm{bend} = \begin{cases} 
V_{\rho_\eps},&\quad \text{mas.}\\
-\tfrac{1}{2}\eps^2 \fL(\pi_M^*g_B,0)\ln\rho_\eps&\quad\text{hol.} 
\end{cases}\,.
\end{equation}
We consider $H^\sV_\cA$ as a fibrewise perturbation of the non-magnetic operator $H^\sV_0 := \rs{H^\sV_\cA}{\cA_\sV^\eps=0}$.
In the hollow case we include potential terms into  $H^\sV_\cA$ that arise from the change of volume measure in the vertical Laplacian and thus do not change the ground state eigenvalue of $H^\sV_0$.

For all $x\in B$ the operator $H^\sV_\cA(x)$ defines a self-adjoint operator on $L^2(M_x)$ with Dirichlet domain $\cD(M_x)$ due to the Kat\={o}-Rellich theorem. The compactness of the fibres $(M_x,g_{M_x})$ yields that $\sigma(H^\sV_\cA(x))$ solely consists of discrete eigenvalues of finite multiplicity accumulating at infinity. 
Let $\lambda_0$ be the ground-state energy of $H^\sV_0$ with a spectral projection $P_0$. That is, for $x\in B$, (cf.~\cite[Section~5.2]{Haag2016a}):
\begin{itemize}
\item \emph{massive quantum waveguides:} \\
The smallest eigenvalue  $\lambda_0>0$ of $\fL(g_\sV, 0)$ with Dirichlet boundary conditions on $M_x$. The associated eigenfunction $\phi_0(x)$ can be chosen real and positive, and $P_0(x)$ projects to the space spanned by this function. 
%
\item \emph{hollow quantum waveguides:} \\
It holds that $\lambda_0\equiv 0$ with ($\eps$-dependent) ground state
\begin{equation}
\label{eq:phi0hol}
\phi_0(x) = \frac{\rho_\eps^{1/2}}{\norm{\rho_\eps}_{L^2(M_x)}} = \Vol_{g_{M_x}}(M_x)^{-1/2} + \cO(\eps).
\end{equation}
\end{itemize}
The latter is the reason why, in the hollow case, we added additional terms to $H^\sV_\cA$. These terms are of order $\eps$, but they have a trivial effect on $\lambda_0$, which would be somewhat obscured if we treat them as a part of the bending potential. 

To first order in perturbation theory, $\lambda_\cA-\lambda_0$ is given by
\[
- 2 \eps^{2-\sigma} \Im\left(\int_{M_x} g_{M_x}\bigl(\dd\phi_0,\cA_\sV^{\eps=0}\phi_0)\vol_{g_{M_x}}\right) = 0.
\]
This equals zero because $\phi_0$ is real-valued. Thus, the asymptotic expansion of the magnetic ground state eigenband reads
\[
\lambda_\cA(x) = \begin{cases} \lambda_0(x) + \cO(\eps^{4-2\sigma}),& \quad  \text{mas.} \\
\cO(\eps^{4-2\sigma}),&\quad\text{hol.} \end{cases}
\]
with $x$-uniform errors due to Definition~\ref{def:QWGbound}. 

The bounded geometry of the waveguide $M$ implies that $\lambda_0(x)$ is separated from the rest of $\sigma(H^\sV_0(x))$ by an $x$-uniform gap~\cite[Proposition~4.1]{Lampart2016}. Thus $\lambda_\cA(x)$ also has a gap, of size $\delta>0$ uniformly in $x$ and $\eps$, for $\eps$ small enough (see \cite[Lemma~5.13]{Haag2016a}). Moreover $P_\cA=P_0+\cO(\eps^{2-\sigma})$ and $\norm{P_\cA(x) \phi_0}_{L^2(M_x)}\geq C$.
We can thus write the entire normalised ground state of $H^\sV_\cA$ as
\[
\phi_\cA(x) := \frac{P_\cA(x)\phi_0(x)}{\norm{P_\cA \phi_0}_{L^2(M_x)}}.
\]
%
This gives a natural choice for the isomorphism of the range of $P_\cA$ with $L^2(B)$, $\Psi\mapsto \scpro{\phi_\cA(x)}{\rs{\Psi}{M_x}}_{L^2(M_x)}$.

In  general, $\Ha^\sigma$ and $\Heff^\sigma$ approximate $H^\sigma$ only on the range of $P_\cA$. If we restrict ourselves to energies below
\begin{equation*}                                                                                                                                               
\Lambda:= \inf_{x\in B} \min(\sigma(H^\sV_0(x))\backslash\{\lambda_0(x)\})                                                                                                                                   \end{equation*}
this restriction is not necessary. This is due to the fact that, for a function $\psi\in L^2(M)$ in the space $\id_{(-\infty, \Lambda]}(H^\sigma)\cH$, the vertical mode $\rs{\psi}{M_x}$ has to concentrate in the  ground state of $H^\sV_\cA$. These considerations lead to the following theorem:

\begin{Theorem} \label{thm:main1}
Let $(M,\cA)$ be a magnetic quantum waveguide of bounded geometry and set $\Lambda:= \inf_{x\in B} \min(\sigma(H^\sV_0(x))\backslash\{\lambda_0(x)\})$. Then for all $N\in\NN$ there exist a projection $P_\cA^\eps$ and a unitary operator $U_\eps^\sigma\in\cL(\cH)\cap\cL(\dom(H^\sigma))$, intertwining $P_\cA$ and $P_\cA^\eps$, such that for every cut-off $\chi\in C^\infty_0((-\infty,\Lambda],[0,1])$, satisfying $\chi^p\in C^\infty_0((-\infty,\Lambda],[0,1])$ for all $p>0$, we have
\[
\norm{H^\sigma \chi(H^\sigma) - U_\eps^\sigma \Heff^\sigma\chi(\Heff^\sigma) \bigl(U_\eps^\sigma\bigr)^\dagger}_\cH = \cO(\eps^N).
\]
In particular, the Hausdorff distance between the spectra of $H^\sigma$ and $\Heff^\sigma$ is small, \iec it holds for every $\delta>0$ that
\[
\dist_\mathrm{H} \bigl(\sigma(H^\sigma)\cap(-\infty,\Lambda-\delta],\sigma(\Heff^\sigma)\cap(-\infty,\Lambda-\delta]\bigr)
\]
is of order $\eps^N$.
\end{Theorem}

\begin{proof}
For $\sigma=0$ this is an immediate consequence of \cite[Corollary~2.3]{Haag2016b}. More presicely (in the notation of the mentioned work), one takes the trivial line bundle $\cE=M\times\CC$ endowed with the connection $\nabla^{M\times\CC}=\dd+\I\pi_M^*\cA_B$ and treats the corrections of the metric and the connection as a perturbation.

For $\sigma=1$ we use the connection $\nabla^{M\times\CC}=\dd+\I\cA_\sH^{\eps=0}$ and treat the higher order corrections again as a perturbation. The contribution of the remaining strong part $\eps^{-1}\pi_M^*\cA_B$ of the magnetic potential may be added to the horizontal Laplacian, since it does not depend on the fibre coordinates (see \cite[Lemma~5.19]{Haag2016a} for the details).
\end{proof}
For $N=1$ we can choose $P_\cA^\eps=P_\cA$, so at first sight the approximation of $H^\sigma$ by $\Ha^\sigma$ (which is much simpler compared to $\Heff^\sigma$ since it does neither incorporate the complicated super-adiabatic projection $P_\cA^\eps$ nor the unitary $U_\eps^\sigma$) yields errors of order $\eps$. More careful inspection shows that for $N>1$ we have $\Ha^\sigma-\Heff^\sigma=\cO(\eps^2)$ as an operator from $W^2(B)$ to $L^2(B)$, so the statement on the spectrum holds for $\Ha^\sigma$ with an error of order $\eps^2$. 

\begin{Corollary} \label{cor:main1}
Let $(M,\cA)$ be a magnetic quantum waveguide of bounded geometry and set $\Lambda:= \inf_{x\in B} \min(\sigma(H^\sV_0(x))\backslash\{\lambda_0(x)\})$. Then there exists a a unitary operator $U_\eps^\sigma\in\cL(\cH)\cap\cL(\dom(H^\sigma))$ such that for every cut-off function $\chi\in C^\infty_0((-\infty,\Lambda],[0,1])$, satisfying $\chi^p\in C^\infty_0((-\infty,\Lambda],[0,1])$ for all $p>0$, we have
\[
\norm{H^\sigma \chi(H^\sigma) - U_\eps^\sigma \Ha^\sigma\chi(\Ha^\sigma) \bigl(U_\eps^\sigma\bigr)^\dagger}_\cH = \cO(\eps^2).
\]
and, in particular, it holds for every $\delta>0$ that
\[
\dist_\mathrm{H} \bigl(\sigma(H^\sigma)\cap(-\infty,\Lambda-\delta],\sigma(\Ha^\sigma)\cap(-\infty,\Lambda-\delta]\bigr)=\cO(\eps^2).
\]
\end{Corollary}

We may further improve the adiabatic approximation if we focus on the low-lying part of the spectrum near $\inf \sigma(H^\sigma)$. 
The behaviour in this energy regime clearly corresponds to that of the low-lying part of the adiabatic operator $\Ha^\sigma$ and therefore is dominated by the characteristics of the unperturbed ground state band $\lambda_0=\min \sigma(H^\sV_0)$:
\begin{itemize}
\item If $x\mapsto \lambda_0(x)$ has a unique non-degenerate minimum $\Lambda_0=\lambda_0(x_0)$ on $B$, the results obtained in \cite{Simon1983} suggest that the leading part of the adiabatic operator behaves like an harmonic oscillator close to $x_0$, schematically
\[
\Ha^\sigma - \Lambda_0 = - \eps^2 \Delta_B + c \bigl(\dist_{g_B }(x,x_0)\bigr)^2 + \dots,\quad c>0,
\]
with eigenvalue spacing of order $\eps$, so the interesting scale for small energies is $\alpha=1$.  In this case there is no spectrum of $\Ha^\sigma$ in the interval $(-\infty, \Lambda_0 + C \eps^2]$, for $\eps$ small enough, which will imply $\sigma(H^\sigma)\cap(-\infty,\Lambda_0+C\eps^2]=\emptyset$.
\item If $\lambda_0(x)= \Lambda_0$ is constant, the adiabatic operator
is given by
\[
\Ha^\sigma - \Lambda_0 =  \eps^2(-\Delta_B + \dots)
\]
and its eigenvalues, if they exist, scale as $\eps^2$. We will see that the latter approximate those of $H^\sigma$ up to errors of order $\eps^4$.
\end{itemize}
More generally we will investigate energies of order $\eps^\alpha$ above the bottom $\Lambda_0:=\inf_{x\in B}\lambda_0(x)$ of the vertical operator for $0<\alpha\leq 2$. The following theorem states that the mutual approximation in this regime is better, compared to Corollary~\ref{cor:main1}, by a factor $\eps^{\alpha/2}$ and even by a factor $\eps^\alpha$ if one merely considers eigenvalues. For moderate magnetic fields ($\sigma=0$) this can be done without further assumptions, for strong magnetic fields ($\sigma=1$) the additional hypothesis that $\cA_B=c^*\cA$ may be gauged away is needed. If $B$ is simply connected, this is equivalent to the vanishing of the pulled-back magnetic field, $\cB_B:= c^* \dd \cA  = 0$. The latter always vanishes if $b=1$. If $\cA_B$ is not exact, then we expect that the adiabatic Hamiltonian $\Ha^\sigma$  needs to be modified by additional terms to achieve the same precision as in the  following theorem.
%
%
%

\begin{Theorem} \label{thm:main2}
Let $(M,\cA)$ be a magnetic quantum waveguide of bounded geometry and $0<\alpha\leq 2$. 
If $\sigma=1$, assume additionally that $\cA_B$ is an exact one-form on $B$.
Then, for every $C>0$, it holds that
\[
\dist_\mathrm{H} \bigl(\sigma(H^\sigma)\cap(-\infty,\Lambda_0+C\eps^\alpha],\sigma(\Ha^\sigma)\cap(-\infty,\Lambda_0+C\eps^\alpha]\bigr)
\]
is of order $\eps^{2+\alpha/2}$. 

If, moreover,  $\Lambda_0+C\eps^\alpha$ is strictly below the essential spectrum of $\Ha^\sigma$, in the sense that for some $\delta>0$ the spectral projection $\id_{(-\infty,\Lambda_0+(C+\delta)\eps^\alpha]}(\Ha^\sigma)$ has finite rank, for $\eps>0$ small enough. Then, if $\mu_0^\sigma < \mu_1^\sigma\leq\dots\leq\mu_K^\sigma$ are all the eigenvalues of $\Ha^\sigma$ below $\Lambda_0+C\eps^\alpha$, $H^\sigma$ has at least $K+1$ eigenvalues $\nu_0^\sigma<\nu_1^\sigma\leq\dots\leq \nu_K^\sigma$ below its essential spectrum and
\[
\abs{\mu_j^\sigma-\nu_j^\sigma} = \cO(\eps^{2+\alpha}),
\]
for $j\in\{0,\dots,K\}$.
\end{Theorem}

\begin{proof}
The statement for the moderate case ($\sigma=0$) is an immediate application of \cite[Proposition 4.14]{Haag2016a} and \cite[Theorem 4.15]{Haag2016a}, where again $\cE=M\times \CC$ and $\nabla^{M\times\CC} = \dd + \I \pi_M^*\cA_B$ (plus higher order corrections that are again treated as a perturbation). The basic idea of this improvement to work is the observation that the super-adiabatic corrections to $\Ha^{\sigma=0}$ (which must also be included in order to obtain a better approximation) essentially consist of horizontal differential operators. But such derivatives $\nabla_{\eps X^\sH}^{\pi_M^*\cA_B}$ for $X\in\sct_\mathrm{b}(\sT B)$ are of order $\eps^{\alpha/2}$ and are therefore small on this $\eps$-dependent energy scale, \iec on the image of $\id_{(-\infty,\Lambda_0+C\eps^\alpha]}(H^{\sigma=0})$ (see \cite[Lemma B.1]{Haag2016a} for the precise statement).
	
For the  strong case ($\sigma=1$) we mention that -- if $\cA_B$ is exact and may be gauged away -- the relevant connection on $M\times\CC$ is given by $\nabla^{M\times\CC}=\dd+\I\cA_\sH^{\eps=0}$ plus higher order corrections and thus has basically the same form as in the moderate case.
\end{proof}

We remark that results can be obtained also for energies higher than $\Lambda$ and projections to other eigenbands than the ground state. The relevant condition is that they are separated from the rest of the spectrum of $H^\sV_\cA$ by a local gap $\delta$, which guarantees regularity of the corresponding spectral projection. As pointed out above, the approximation of spectra is not mutual as for low energies, but there is always spectrum of $H^\sigma$ near that of 
$\Heff^\sigma$ (see~\cite[Corollary~1.2]{Haag2016b}).

\section{Moderate Magnetic Fields (\texorpdfstring{$\sigma=0$}{sigma = 0})} \label{chap:unscaled}

This section is dedicated to the calculation of the adiabatic 
operator~\eqref{eq:Ha} for the case of moderate magnetic fields. To simplify the 
notation, we will frequently drop the superscript $\sigma$. Since we may apply 
both Corollary~\ref{cor:main1} and Theorem~\ref{thm:main2} without further 
restrictions, we treat both cases at once  and aim for the approximation of the 
spectrum of $H$ by that of $\Ha$ for energies of the order $\eps^\alpha$ above 
$\Lambda_0=\inf \lambda_0$, for $\alpha\in[0,2]$.

We need to compute $\Ha$ up to errors of order $\eps^{2+\alpha}$ on the 
appropriate energy scale. This is implemented by estimating the errors in 
$\cL(K_\alpha(B),L^2(B))$, with the $\eps$ and $\alpha$-dependent spaces
\[
K_\alpha(B):= \im\bigl(\chi_{[\Lambda_0,\Lambda_0+C\eps^\alpha]}(-\eps^2\Delta_{g_B}+\lambda_0)\bigr)\subset L^2(B),\quad \alpha\in[0,2]
\]
equipped with the restriction of the $L^2$-norm, where $C$ is an arbitrary constant for $\alpha>0$ and $C<\Lambda-\Lambda_0$ for $\alpha=0$.
Note that $K_{\alpha_1}(B)\subset K_{\alpha_2}(B)$ is a closed subspace for 
$\alpha_1>\alpha_2$.

We will see below (see also~\cite[Theorem 2.5]{WaTe2013}) that the 
non-magnetic adiabatic operator may be written in the form
\begin{equation}\label{eq:Ha_nm}
 \Ha^\mathrm{nm} = \eps^2 \underline{\fL} \bigl(\expval{ G^\eps_\mathrm{hor}}_0, 
0\bigr)+ \lambda_0 + \eps^2 V_\mathrm{BH}+ \eps^2 \expval{V _\mathrm{bend}}_0
\end{equation}
%
%
%
The quantities in this equation are the following:
\begin{itemize}
 \item The Laplacian $\underline{\fL}$ on the base $B$ which is defined by substituting the
 non-magnetic effective metric, given for $t_1,t_2\in \sT_x B$ by
\begin{align}
\label{eq:Geff}
\langle G^\eps_\mathrm{hor} \rangle_0(t_1,t_2):=&
\scpro{\phi_0}{G_\mathrm{hor}^\eps(t_1^\sH,t_2^\sH)\phi_0}_{L^2(M_x)} \\
=& g_B(t_1,t_2) + \eps 
\scpro{\phi_0}{h^\eps(t_1^\sH,t_2^\sH)\phi_0}_{L^2(M_x)},\nonumber
\end{align}
into the quadratic from
\begin{equation*}
 \langle \psi, \underline{\fL}(\fg,\fA) \psi \rangle_{L^2(B)}
=\int_B \fg
\bigl(\overline{\nabla^\fA\psi},
\nabla^\fA \psi\bigr)  \vol_{g_B}.
\end{equation*}
The underline $\underline{\fL}$ is used here to distinguish the so-defined operator on the base from its analogue $\fL$ on $M$, defined by~\eqref{eq:Deltahor}.
\item The non-magnetic Born-Huang potential is given by
\begin{equation*}
V_\mathrm{BH}(x) :=  \int_{M_x} G_\mathrm{hor}^\eps(\dd \phi_0, \dd \phi_0) 
 - \mathrm{div}_g \bigl(\phi_0 G_\mathrm{hor}^\eps(\dd \phi_0, \cdot ) \bigr)
\vol_{g_{M_x}}.
\end{equation*}
Equivalently, it may be defined by its quadratic form
\begin{align}\label{eq:VBH}
  \scpro{\psi}{V_\mathrm{BH} \psi}_{L^2(B)}    =& \int_{M} \abs{\pi_M^*\psi}^2  
G_\mathrm{hor}^\eps(\dd \phi_0, \dd \phi_0)  \\
&  + 2 \Re \bigl(\phi_0 (\pi_M^*\psi)
G_\mathrm{hor}^\eps(\dd \pi_M^*\overline{\psi}, \dd \phi_0)\bigr)\vol_g.\nonumber
\end{align}
We emphasise that this defines a potential, and not a derivation, due 
to the real part.   
\item The averaged bending potential is
\begin{equation*}
 \expval{V _\mathrm{bend}}_0(x):= \scpro{\phi_0}{V _\mathrm{bend} \phi_0}_{L^2(M_x)}.
\end{equation*}
\end{itemize}


\begin{Remark}
The notation $\expval{\cdot}_0$ is used here to denote averaging over the fibre with weight $\abs{\phi_0}^2$, we will later use the notation $\expval{\cdot}_\cA$ when the weight is given by $\abs{\phi_\cA}^2$.
\end{Remark}

\begin{Remark} \label{rem:VBH}
	The leading order of the non-magnetic Born-Huang potential $V_\mathrm{BH}$ is given by
	\begin{align*}
	&\int_{M_x} \pi_M^*g_B(\dd \phi_0, \dd \phi_0) 
	- \mathrm{div}_g \bigl(\phi_0\, \pi_M^*g_B(\dd \phi_0, \cdot ) \bigr)
	\vol_{g_{M_x}} \\
	&\ = \int_{M_x} \pi_M^*g_B(\dd \phi_0, \dd \phi_0)\vol_{g_{M_x}} - \dive_{g_B} \int_{M_x} \phi_0\, \pi_M^*g_B(\dd \phi_0, \cdot )\vol_{g_{M_x}}
	\end{align*}
	and equals the adiabatic potential found in \cite[Equation~(12)]{Haag2015}. Here, we used the facts that $\phi_0$ vanishes on the boundary $\partial M_x$ and that $ \pi_M^*g_B(\dd \phi_0, \cdot )$ is a horizontal vector field. Moreover, it is shown in \cite[Section 5.2]{Haag2015} that the second integral coincides with $\tfrac{1}{2}\expval{\eta}_0$, where $\eta\in\sct(\sH M)$ stands for the mean curvature vector of the fibres $M_x\hookrightarrow M$. Thus, if we deal with massive quantum waveguides (for which $\eta$ vanishes), the non-magnetic Born-Huang potential has the expansion
	\[
	V_\mathrm{BH}(x) = \int_{M_x} \hspace{-2pt}\pi_M^*g_B(\dd \phi_0, \dd \phi_0)\vol_{g_{M_x}} + \,\cO(\eps) =: \norm{\dd\phi_0}^2_{L^2(\sH M^*\vert_{M_x})} + \,\cO(\eps).
	\]
\end{Remark}

In order to simplify later calculations, we establish the following Lemma:
\begin{Lemma}\label{lem:Berry triv}
 Let $\phi\in \{\phi_0, \phi_\cA\}$, $\nabla$ be a connection on the trivial line bundle $B\times \CC$, $\varphi, \psi\in W^1(B)$ and $\xi \in C^\infty(\sT^*B \otimes \CC)$ and $\mu \in C^\infty(\sT^*M \otimes \CC)$  complex one-forms.
 Then the quadratic form
\begin{align*}
Q_\xi(\phi, \psi)= 
\int_M & G^\eps_{\mathrm{hor}}\left(\overline{\phi\pi_M^*(\nabla +\xi )\varphi}, \phi\pi_M^*(\nabla +\xi )\psi)\right)\\
& + G^\eps_{\mathrm{hor}}\left(\overline{(\mu - \phi\pi_M^*\xi )\pi_M^*\varphi}, \phi\pi_M^*(\nabla +\xi )\psi)\right)\\
&+ G^\eps_{\mathrm{hor}}\left(\overline{\phi\pi_M^*(\nabla +\xi )\varphi}, (\mu - \phi\pi_M^*\xi )\pi_M^*\psi)\right)\\
&+ G^\eps_{\mathrm{hor}}\left(\overline{(\mu - \phi\pi_M^*\xi )\pi_M^*\varphi}, (\mu - \phi\pi_M^*\xi )\pi_M^*\psi)\right)\vol_g
\end{align*}
is independent of $\xi$.
\end{Lemma}
\begin{proof}
 Expanding the expression, one sees that the terms containing $\xi$ cancel each other out, and $Q_\xi=Q_0$ for all $\xi$.
\end{proof}

The adiabatic operator for moderate magnetic fields is obtained from the non-magnetic adiabatic operator $\Ha^\mathrm{nm}$ by minimal coupling to an effective magnetic potential $\cA_\mathrm{eff}^\eps$:

\begin{Theorem} \label{thm:weak}
The adiabatic operator for moderate magnetic fields is given by the non-magnetic adiabatic operator, minimally coupled to the effective magnetic potential $\cA_\mathrm{eff}^\eps := \cA_B + \eps \cA_1$, where $\cA_B:=c^*\cA$ and 
\[
\cA_1(X):=\expval{\cA_\sH^{\eps=0}(X^\sH)}_0
\]
for $X\in\sct(\sT B)$, up to errors of order $\eps^{2+\alpha}$ in $\cL(K_\alpha(B),L^2(B))$ for $\alpha\in[0,2]$.
That is, 
\begin{equation*}
 \Ha^{\sigma=0} = \eps^2 \underline{\fL} \bigl(\expval{G^\eps_\mathrm{hor}}_0, \cA_\mathrm{eff}^\eps\bigr) +  \lambda_0 + 
 \eps^2 V_\mathrm{BH}+ \eps^2 \expval{V _\mathrm{bend}}_0 + \cO(\eps^{2+\alpha}).
\end{equation*}
\end{Theorem}

\begin{proof}
We proceed by expanding the quadratic form of $\Ha^\sigma$ and bounding the error terms in the appropriate norm.
By definition, we have 
\begin{equation*}
\langle \varphi, \Ha^\sigma \psi \rangle_{L^2(B)}
=\scpro{(\pi_M^*\varphi) \phi_\cA}{ H^\sigma \phi_\cA (\pi_M^*\psi)}_{L^2(M)},
\end{equation*}
where we choose $\varphi\in W^1(B)$ and $\psi\in K_\alpha(B)$.
Since $\lambda_\cA=\lambda_0 + \cO(\eps^4)$ and $\phi_\cA=\phi_0+\cO(\eps^2)$ for weak magnetic fields, we immediately see that 
\[
\eps^2 \langle V_\mathrm{bend}\rangle_\cA = \eps^2\langle V_\mathrm{bend}\rangle_0+\cO(\eps^4).
\]

It now remains to analyse the quadratic form of the projected horizontal 
Laplacian. We need to estimate the errors by $\eps^{2+\alpha} \norm{\varphi}_{L^2(B)} \norm{\psi}_{K_\alpha(B)}$. To achieve this, we will make use of the fact (see Equation~\eqref{eq:Kalpha est} below) that $\norm{\eps \dd \psi}_{L^2(\sT^* B)}=\cO(\eps^{\alpha/2})$ for $\psi\in K_\alpha(B)$. When derivatives act on $\varphi$, we integrate by parts, so that they act on $\psi$ and we can use the same estimate. We will not perform this step in detail, but just remark that integration by parts is always justified due to the factors of $\phi_0$ or $\phi_\cA$ in the integrals, which vanish on the boundary of a massive waveguide.

We split the quadratic form using the projection $P_\cA$:
\begin{subequations}
\label{eq:P}
\begin{align}
& \scpro{(\pi_M^*\varphi) \phi_\cA, \eps^2 \fL(G_\mathrm{hor}^\eps}{\cA_\mathrm{hor}^\eps) (\pi_M^*\psi) \phi_\cA}_{L^2(M)}\nonumber \\
%
&\ = \int_M G_\mathrm{hor}^\eps\Bigl(\overline{P_\cA \eps \nabla^{\cA_\mathrm{hor}^\eps}
(\pi_M^*\varphi)\phi_\cA},P_\cA \eps \nabla^{\cA_\mathrm{hor}^\eps}(\pi_M^*\psi)\phi_\cA\Bigr)\vol_g \label{eq:P1} \\
&\ \hphantom{=}\ + \int_M G_\mathrm{hor}^\eps\Bigl(\overline{P^\bot_\cA\eps\nabla^{\cA_\mathrm{hor}^\eps} (\pi_M^*\varphi)\phi_\cA},P^\bot_\cA\eps\nabla^{\cA_\mathrm{hor}^\eps}(\pi_M^*\psi)\phi_\cA\Bigr)\vol_g \label{eq:P2} \\
&\ \hphantom{=}\ + \int_M G_\mathrm{hor}^\eps\Bigl(\overline{P_\cA\eps\nabla^{\cA_\mathrm{hor}^\eps} (\pi_M^*\varphi)\phi_\cA},P^\bot_\cA\eps\nabla^{\cA_\mathrm{hor}^\eps}(\pi_M^*\psi)\phi_\cA\Bigr)\vol_g \label{eq:P3} \\
&\ \hphantom{=}\ + \int_M G_\mathrm{hor}^\eps\Bigl(\overline{P^\bot_\cA\eps\nabla^{\cA_\mathrm{hor}^\eps} (\pi_M^*\varphi)\phi_\cA},P_\cA\eps\nabla^{\cA_\mathrm{hor}^\eps}(\pi_M^*\psi)\phi_\cA\Bigr)\vol_g \label{eq:P4},
\end{align}
\end{subequations}
where we used the abbreviation $\nabla^{\cA_\mathrm{hor}^\eps}=\dd + \I \cA_\mathrm{hor}^\eps$.

We start by simplifying this expression  using Lemma~\ref{lem:Berry triv}. We have
\begin{multline}
 P_\cA\eps \nabla^{\cA_\mathrm{hor}^\eps}(\pi_M^*\psi)\phi_\cA   =\pi_M^*\left(\eps\big(\dd + \I \expval{\cA_{\mathrm{hor}}^\eps}_\cA 
+\scpro{\phi_\cA}{\dd\phi_\cA}_{L^2(M_x)}\big)\psi \right)\phi_\cA \notag
\end{multline}
and, because $[\cA_B(X),P_\cA]=0$ for any vector field,
\begin{multline}
\notag  
P^\bot_\cA \eps\nabla^{\cA_\mathrm{hor}^\eps}(\pi_M^*\psi)\phi_\cA = \eps\Bigl(\dd\phi_\cA -  \scpro{\phi_\cA}{\dd\phi_\cA}_{L^2(M_x)} + \I\eps P_\cA^\bot \cA_\sH^{\eps}\phi_\cA\Bigr)(\pi_M^*\psi).
\end{multline}
We can thus apply Lemma~\ref{lem:Berry triv} with $\nabla=\dd + \I \langle \cA_{\mathrm{hor}}^\eps \rangle_\cA$, $\xi=\scpro{\phi_\cA}{\dd\phi_\cA}_{L^2(M_x)}$ and $\mu=\dd\phi_\cA+ \I\eps P_\cA^\bot \cA_\sH^{\eps}\phi_\cA$, which means that we can continue as if $\xi=0$.
In particular, this proves the formula~\eqref{eq:Ha_nm} when $\cA=0$.

We now show that, in Equation~\eqref{eq:P1} with $\xi=0$, $\expval{ \cA_{\mathrm{hor}}^\eps}_\cA$ can be replaced by $\cA_\mathrm{eff}^\eps$ and $\langle G_\mathrm{hor}^\eps \rangle_\cA$ with $\expval{G_\mathrm{hor}^\eps}_0$, up to errors of order $\eps^{3+\alpha/2}$.
To achieve this, observe that
\begin{align*}
\int_B \eps^2 g_B(\overline{\dd\psi}, \dd\psi)\vol_{g_B} \nonumber
 &\leq \scpro{\psi}{\bigl(-\eps^2\Delta_{g_B} +(\lambda_0-\Lambda_0)\bigr)\psi}_{L^2(B)}\\
  &\leq C \eps^{\alpha} \norm{\psi}_{K_\alpha(B)}^2,
\end{align*}
which implies that
\begin{equation}\label{eq:Kalpha est}
 \int_B \eps^2 g_B\Bigl(\overline{\nabla^{\cA_\mathrm{eff}^\eps}\psi}, \nabla^{\cA_\mathrm{eff}^\eps}\psi\Bigr)\vol_{g_B}=\cO\bigl(\eps^\alpha \norm{\psi}_{K_\alpha}^2\bigr).
\end{equation}
Now expand
\begin{equation*}
 \expval{\cA_{\mathrm{hor}}^\eps}_\cA
 = \cA_B + \eps \expval{\cA_\sH^{\eps}}_\cA
 =\cA_\mathrm{eff}^\eps + \cO(\eps^3)
\end{equation*}
 and
 \begin{equation*}
  \langle G_\mathrm{hor}^\eps \rangle_\cA=g_B + \eps \langle h^\eps 
\rangle_\cA= \langle G_\mathrm{hor}^\eps \rangle_0 + \mathcal{O}(\eps^3)
 \end{equation*}
(where the error is measured by $g_B$). 
Together with Equation~\eqref{eq:Kalpha est} and integration by parts this shows that
\begin{align*}
&  \int_M \abs{\phi_\cA}^2 G_\mathrm{hor}^\eps
\Bigl(\eps\pi_M^*\overline{\nabla^{\expval{\cA_{\mathrm{hor}}^\eps}_\cA}\varphi}, 
\eps\pi_M^*\nabla^{\expval{\cA_{\mathrm{hor}}^\eps}_\cA}\psi\Bigr)\vol_g \\
& =\ \scpro{\varphi}{\underline{\fL}(\expval{G_\mathrm{hor}^\eps}_0,\cA_\mathrm{eff}^\eps)\psi}_{L^2(B)}
+ \cO\bigl(\eps^{3+\alpha/2}\norm{\varphi}_{L^2(B)}\norm{\psi}_{K_\alpha(B)}\bigr).
\end{align*}

For the remaining terms~\eqref{eq:P2}--\eqref{eq:P4}, we use the expansion 
\begin{equation*}
 \dd \phi_\cA + \eps P_\cA^\perp \cA_\sH^\eps \phi_\cA= \dd \phi_0 + \eps P_0^\perp \cA_\sH^{\eps=0} \phi_0 + \cO(\eps^2),
\end{equation*}
which yields, by the same reasoning as above and neglecting $\xi$, 
\begin{align}
 &\eqref{eq:P2} + \eqref{eq:P3}\nonumber \\
 &\ \cong \int_M  G_\mathrm{hor}^\eps \Bigl(\eps\overline{\phi_0 \pi_M^*\nabla^{\cA_{\mathrm{eff}}^\eps}\varphi}, \eps(\dd\phi_0 +\I \eps P_0^\bot \cA_\sH^{\eps=0}\phi_0)\pi_M^*\psi \Bigr)\vol_g  \label{eq:P23_new} \\
 &\ \hphantom{\cong}\ +\int_M  G_\mathrm{hor}^\eps \Bigl(\eps\overline{ (\dd\phi_0 +\I \eps P_0^\bot \cA_\sH^{\eps=0}\phi_0 )\pi_M^*\varphi},\eps \phi_0\pi_M^*\nabla^{\cA_{\mathrm{eff}}^\eps}\psi \Bigr)\vol_g \nonumber \\
 &\ \hphantom{\cong}\ +\cO\bigl(\eps^{3+\alpha/2}\norm{\varphi}_{L^2(B)}\norm{\psi}_{K_\alpha(B)}\bigr). \nonumber
\end{align}
The terms involving $\I \eps P_0^\bot \cA_\sH^{\eps=0}\phi_0$ are also of order $\eps^{3+\alpha/2}$, because 
\begin{align*}
& \int_{M_x} (\pi_M^*g_B)\left(\overline{\phi_0 \pi_M^*\nabla^{\cA_{\mathrm{eff}}^\eps}\varphi}, P_0^\bot \cA_\sH^{\eps=0}\phi_0\right)\vol_{g_{M_x}} \\
&\  =g_B\Bigl(\overline{\nabla^{\cA_{\mathrm{eff}}^\eps}\varphi}, \scpro{\phi_0}{P_0^\bot \cA_\sH^{\eps=0}\phi_0}_{L^2(M_x)} \Bigr)=0.
\end{align*}
Furthermore, the remaining terms in Equation~\eqref{eq:P23_new} that contain $\I \cA_\mathrm{eff}$ yield
\begin{equation}\label{eq:A odd}
 \int_M  -2\pi_M^*\overline{\varphi}\psi G_\mathrm{hor}^\eps \bigl(\Im(\phi_0 \cA_{\mathrm{eff}}^\eps), \dd\phi_0 \bigr)\vol_g=0 
\end{equation}
because $\phi_0$  and $\cA_{\mathrm{eff}}^\eps$ are real. The relevant contribution of~\eqref{eq:P2},~\eqref{eq:P3} thus reduces to
\begin{multline}\label{eq:P23_final}
  \eps^2 \int_M G_\mathrm{hor}^\eps\bigl(\phi_0 \pi_M^*\overline{\dd\varphi}, (\dd\phi_0)\pi_M^*\psi \bigr)
 + \int_M G_\mathrm{hor}^\eps\bigl((\dd\phi_0)\overline{\pi_M^*\varphi}, \phi_0 \pi_M^*\dd\psi \bigr) \vol_g.
\end{multline}

The contribution of the last term,~\eqref{eq:P4}, is just
\begin{align}\label{eq:P4_new}
 \eps^2 \int_M G_\mathrm{hor}^\eps (\dd\phi_0 , \dd\phi_0)\pi_M^*(\overline{\varphi}\psi) \vol_g  + \cO\bigl(\eps^4\norm{\varphi}_{L^2(B)}\norm{\psi}_{L^2(B)}\bigr), \nonumber
\end{align}
as the imaginary part of  $G_\mathrm{hor}^\eps(\dd\phi_0 ,  P_0^\bot \cA_\sH^{\eps=0}\phi_0)$ vanishes.
Combining this with~\eqref{eq:P23_final} gives exactly the quadratic form of $V_\mathrm{BH}$.
\end{proof}

We remark that this proof shows that the error is actually of order~$\eps^{3+\alpha/2}$, which is slightly better than $\eps^{2+\alpha}$. But this is not very relevant, as the error of the original approximation of $H$ by $\Ha$ is $\eps^{2+\alpha}$.

\section{Strong Magnetic Fields (\texorpdfstring{$\sigma=1$}{sigma = 1})} \label{chap:scaled}

In this section we compute the adiabatic operator~\eqref{eq:Ha}. We will again drop the superscript $\sigma$, most of the time.

\begin{Theorem}\label{thm:strong}
The adiabatic operator for $\sigma=1$ is given by
\[
\Ha^{\sigma=1} = \eps^2 \underline{\fL}\bigl(\expval{G^\eps_\mathrm{hor}}_\cA, \eps^{-1} \cA_B+ \cA_P\bigr)
 + \lambda_\cA  +  \eps^2 \expval{V_\mathrm{bend}}_\cA + \eps^2 V_\mathrm{BH}^\cA + \eps^2 D_\mathrm{mag},
\]
with the abbreviations $\cA_P:=\expval{\cA_\sH^\eps}_\cA$, $\cA_{P^\perp}:=P_\cA^\perp \cA_\sH^\eps \phi_\cA$, where  the magnetic Born-Huang potential is given by
\begin{align*}
V_\mathrm{BH}^\cA (x) & :=  \int_{M_x} G_\mathrm{hor}^\eps\bigl(\overline{\dd \phi_\cA+ \I \cA_{P^\perp}}, \dd \phi_\cA+ \I \cA_{P^\perp}\bigr) \\
&\hphantom{:=  \int_{M_x}}\  - \dive_g \bigl(\overline{\phi_\cA} G_\mathrm{hor}^\eps(\dd \phi_\cA+ \I \cA_{P^\perp}, \cdot ) \bigr)
\vol_{g_{M_x}},
\end{align*}
and the first-order differential operator
\[
(D_\mathrm{mag}\psi)(x) =2 \hspace*{-4pt}\int\limits_{M_x} \hspace*{-4pt} G_\mathrm{hor}^\eps\Bigl(\Im \bigl(\overline{\phi_\cA}(\dd \phi_\cA+ \I \cA_{P^\perp})\bigr),  \pi_M^*(-\I\dd + \eps^{-1} \cA_B+ \cA_P) \psi \Bigr)\vol_{g_{M_x}}.
\]
\end{Theorem}
This statement is exact, with no further errors, but the terms in the adiabatic operator may of course be expanded using perturbation theory for $\lambda_\cA$ and $\phi_\cA$, as well as an expansion of $\cA_\sH^\eps$. In particular, when $\cA_B$ is an exact form, Theorem~\ref{thm:main2} applies and one can obtain a more explicit form of the effective Hamiltonian on $\eps^\alpha$-energy scales, cf.\ Section~\ref{sec:652}.  

%
%
\begin{proof}
The potentials $\lambda_\cA$ and $\eps^2 \expval{V_\mathrm{bend}}_\cA$ are  the exact projections of $H_\sV^\cA$ and the bending term, so there is nothing to prove here and we are again left to calculate the projection of $\fL\left( G_\mathrm{hor}^\eps, \eps^{-1}\cA_\mathrm{hor}^\eps \right)$ with $P_\cA$. 

As in the proof of Theorem~\ref{thm:weak}, we start by applying Lemma~\ref{lem:Berry triv}, for which we need
\begin{align}
& P_\cA\eps \nabla^{\eps^{-1} \cA_\mathrm{hor}^\eps}(\pi_M^*\psi)\phi_\cA \nonumber \\
&\ = \pi_M^*\Bigl(\bigl(\eps \dd + \I \cA_B +  \eps \I \cA_P
+\eps \scpro{\phi_\cA}{\dd\phi_\cA}_{L^2(M_x)}\bigr)\psi \Bigr)\phi_\cA \label{eq:aux1_strong}
\end{align}
and
\[
P^\bot_\cA \eps\nabla^{\eps^{-1}\cA_\mathrm{hor}^\eps}(\pi_M^*\psi)\phi_\cA
= \eps\bigl(\dd\phi_\cA -  \scpro{\phi_\cA}{\dd\phi_\cA}_{L^2(M_x)}+ \I \cA_{P^\perp}\bigr)(\pi_M^*\psi).
\]

Writing out the quadratic form of $\eps^2 P_\cA \fL( G_\mathrm{hor}^\eps, \eps^{-1}\cA_\mathrm{hor}^\eps)P_\cA$ as in~\eqref{eq:P1} -- \eqref{eq:P4}, we can thus neglect $\xi=\scpro{\phi_\cA}{\dd\phi_\cA}_{L^2(M_x)}$ by Lemma~\ref{lem:Berry triv}.
The term with the projection $P_\cA$ applied twice then yields exactly the quadratic form of $\eps^2 \underline{\fL}(\expval{ G^\eps_\mathrm{hor}}_\cA, \eps^{-1} \cA_B+ \cA_P)$.

Using that $\cA_P$ is real and integration by parts, one sees that the terms with exactly one $P_\cA^\perp$ yield the quadratic form of the operator 
\[
\int_{M_x} \dive_g \bigl(\overline{\phi_\cA} G_\mathrm{hor}^\eps(\dd \phi_\cA+ \I \cA_{P^\perp}, \cdot ) \bigr)\vol_{g_{M_x}} + D_\mathrm{mag}.
\]
Note that this operator is self-adjoint but $D_\mathrm{mag}$ and the complex potential separately are not.
The first term in this expression is completed by the one with two occurrences of $P_\cA^\perp$ to the potential $V^\cA_{\mathrm{BH}}$.
\end{proof}

In view of Corollary~\ref{cor:main1} there is no additional restriction on $\cA$ and we need to calculate $\Ha$ up to errors of order $\eps^2$:

\begin{Corollary}[$\alpha=0$] \label{cor:strong}
The adiabatic operator for strong magnetic fields is given by
\[
\Ha^{\sigma=1} = \eps^2 \underline{\fL}\bigl(\expval{G^\eps_\mathrm{hor}}_\cA, \eps^{-1}\cA^\eps_{\mathrm{eff}}\bigr) + \lambda_0  +  \cO(\eps^2)
\]
with errors of order $\eps^2$ in $\cL(K_0(B),L^2(B))$, where the effective magnetic potential $\cA_\mathrm{eff}^\eps$ is the same as in Theorem~\ref{thm:weak}.
\end{Corollary}

\begin{proof}
We have $\lambda_\cA=\lambda_0 + \cO(\eps^2)$ and $\eps^2\expval{ V_{\mathrm{bend}}}_\cA=\cO(\eps^2)$, so it remains to show that $\eps^2 D_\mathrm{mag}=\cO(\eps^2)$, and that $\cA_P$ can be replaced by $\expval{\cA_\sH^{\eps=0}}_0$, with appropriate errors.
The latter is a consequence of Equation~\eqref{eq:aux1_strong}, as $\eps(\cA_P - \expval{\cA_\sH^{\eps=0}}_0)=\cO(\eps^2)$.
 
To see that $D_\mathrm{mag}$ is bounded uniformly in $\eps$ from $K_0(B)$ to $L^2(B)$, it is sufficient to show that
\[
\int_{M_x} G_\mathrm{hor}^\eps\Bigl(\Im\bigl(\overline{\phi_\cA}(\dd \phi_\cA+ \I \cA_{P^\perp})\bigr), \cdot \Bigr)\vol_{g_{M_x}}=\cO(\eps), 
\]
as a one-from on $(B,g_B)$ (cf.~Equation~\eqref{eq:Kalpha est}). This holds because $\phi_\cA=\phi_0 + \cO(\eps)$ with real $\phi_0$ and Equation~\eqref{eq:A odd}.
\end{proof}

\section{Application to Quantum Tubes} \label{chap:conv}

This section is devoted to the explicit computation of the magnetic adiabatic operators of massive or hollow quantum tubes, that is waveguides modelled along closed or infinite open curves in $\RR^3$. We will discuss the leading contribution of the magnetic effects for moderate   and  strong magnetic fields in Section~\ref{sec:unscaled} and Section~\ref{sec:scaled}, respectively.

\subsection{The Geometry} \label{subsec:geom}

Let $c:B\to \RR^3$ be either an infinite curve, \iec $B=\RR$, or a closed curve of length $2L$, \iec $B=\RR/2L\ZZ \cong [-L,L)$, that is smoothly embedded in $(\RR^3,\updelta)$, bounded with all its derivatives and parametrised by arc length, so $g_B=\dd x \otimes \dd x$. In both cases the normal bundle $\sN B$ is diffeomorphic to the trivial bundle $B\times \RR^2$. For the following explicit computations we choose coordinates induced by   parallel orthonormal frames in $\sN B$. While for $B=\RR$ such frames exist globally, for $B=\RR/2L\ZZ$ they might be discontinuous at one point. Hence, in this case, we restrict the computations to $B\setminus \{x_0\}\cong (-L,L)$. In order to not overburden the notation, we do not make this explicit in every step, but keep in mind that $B$ needs to be replaced by $B\setminus \{x_0\}$ whenever necessary. The relevant geometric objects will always be well defined on all of $B$, see for example Remark~\ref{rem:independent}.

Explicitly, a parallel orthonormal frame can be constructed as follows: At $ 0\in B$ we pick an orthonormal basis $\{\tau_0,e_{1,0},e_{2,0}\}$ of $\sT_{c(0)}\RR^3\cong\RR^3$ such that $\tau_0=c'(0)$ is tangent and~$(e_{1,0},e_{2,0})$ are normal to the curve at $c(0)$. One then obtains 
a frame $(\tau,e_1,e_2)$ of $\rs{\sT \RR^3}{c(B)}$ as the solution of the coupled system of differential equations 
\[
\frac{\dd}{\dd x}\begin{pmatrix} \tau(x) \\ e_1(x) \\ e_2(x) \end{pmatrix} = \begin{pmatrix} 0 & \kappa_1(x) & \kappa_2(x) \\ -\kappa_1(x) & 0 & 0 \\ -\kappa_2(x) & 0 & 0 \end{pmatrix}
\begin{pmatrix} \tau(x) \\ e_1(x) \\ \e_2(x) \end{pmatrix}
\]
with initial data $(\tau( 0),e_1( 0),e_2( 0))=(\tau_0,e_{1,0},e_{2,0})$, 
where
\[
\kappa_j: B\to \RR,\quad x\mapsto \updelta\bigl(c''(x),e_j(x)\bigr)\text{ for $j\in\{1,2\}$}
\]
are the mean curvatures of the curve, see for instance \cite{Bishop1975}. This is an orthonormal frame of $(\rs{\sT\RR^3}{c(B)},\updelta)$ with $\tau(x)=c'(x)$ for all $x\in B$ which we identify with an orthonormal frame of $\sT B\oplus \sN B$ by means of~\eqref{eq:TBNB}. Put differently, this frame is obtained by the parallel transport of the orthonormal basis $(\tau_0,e_{1,0},e_{2,0})$  along the curve $B$ with respect to the induced connection $c^*\nabla^\updelta$. Consequently, $x\mapsto\tau(x)=\partial_x$ is a trivialisation of the tangent bundle $\sT B$, whereas $x\mapsto (e_1(x),e_2(x))$ yields bundle coordinates  
 \[
\Xi: B \times \RR^2 \to \sN B,\quad (x,n_1,n_2) \mapsto n_1 e_1(x)+n_2 e_2(x)  
\]
that come along with coordinate vector fields $\partial_x^\#$, $\partial_{n_1}$ and $\partial_{n_2}$, where $\partial_x^\#$ is the horizontal lift (which coincides with the product lift) of ${\partial_x} $ to $\sN B$.

%

%
%

  \subsection{The Pullback of the Riemannian Metric}

As an intermediate result we obtain a family of pullback metrics $\{G_\mathrm{int}^\eps:=\Phi_\eps^*(\eps^{-2}\updelta)\}_{0<\eps\leq 1}$ on $\sN B^r:= \Xi(B\times \mathbb{B}^2_{r}(0)) $ induced by
\[
\Phi_\eps : (x,n_1,n_2)\mapsto c(x) + \eps \bigl(n_1 e_1(x) + n_2 e_2(x)\bigr)\,.
\]
 It is given by (see \cite[Example 5.6]{Haag2016a})
\begin{align*}
\rs{G_\mathrm{int}^\eps}{(x,n)} &= \eps^{-2}\bigl(1 -\eps \scpro{n}{\kappa(x)}_{\RR^2}\bigr)^2\, \dd x^\#\hspace{-2pt} \otimes\hspace{-1pt} \dd x^\#   + \dd n_1\hspace{-2pt}\otimes\hspace{-1pt} \dd n_1 + \dd n_2\hspace{-2pt}\otimes\hspace{-1pt} \dd n_2,
\end{align*}
where we used the notation $\dd x^\# := \pi_{\sN B}^*\dd x$.  

\begin{Remark}\label{rem:independent}
Note that even in the case, where our coordinates are defined only on $(B\setminus\{x_0\})\times \RR^2$, the metric $G_\mathrm{int}^\eps$ is well defined on all of $\sN B^r$. This is because the bases $(e_1(-L), e_2(-L))$ and $(e_1(L), e_2(L))$ differ at most by a transformation in $\mathrm{SO}(2)$, while the appearing expressions are invariant under rotations of the fibres. 
\end{Remark}
%

The final form of the family $\{G^\eps:=\varpi^*G_\mathrm{int}^\eps\}_{0<\eps\leq 1}$ now depends on the concrete realisation of the waveguide induced by the embedding $\varpi:M\to\sN B^r$. Note that for $B=\RR$ we can without loss of generality pick $M=\RR\times F$, while for $B=\RR/2L\ZZ$  the bundle $M\xrightarrow{\pi_M} B$ need not be trivial. As we will see in the following example, the global structure of $M$ is then encoded in the map $\varpi$.

\begin{Example} \label{ex:geoconv}
Let us consider two illustrative choices of the embedding $\varpi$ (see \cite[Example~5.7]{Haag2016a} and \cite[Example~5.9]{Haag2016a}):
\begin{enumerate}
\item \emph{massive quantum tube, $\dim(F)=2$:} \\
Let $F\subset\RR^2$ be a connected (in general not rotationally invariant) bounded domain and $R>0$ such that $F\subset\mathbb{B}_R(0)$. Then for $y=(y_1,y_2)\in F$ the mapping
\[
\varpi :(x,y_1,y_2) \mapsto \bigl(\ell(x)\fr(x)y\bigr)_1 e_1(x) + \bigl(\ell(x)\fr(x)y\bigr)_2 e_2(x),
\]
with
\[
\fr:  x \mapsto \begin{pmatrix} \cos\bigl(\varphi(x)\bigr) & - \sin\bigl(\varphi(x)\bigr) \\ \sin\bigl(\varphi(x)\bigr) & \cos\bigl(\varphi(x)\bigr) \end{pmatrix}
\]
models a waveguide whose cross-section is the fixed set $F$ that twists around the curve with respect to an $x$-dependent angle $\varphi:B\to \RR$   and scales with an $x$-dependent factor $\ell\in C^\infty_\mathrm{b}(B,[l_-,l_+])$ such that $0<l_- < l_+ < r/R$.  For $B=\RR$ we require that $\varphi\in\sct_\mathrm{b}(\RR)$ and for $B=\RR/2L\ZZ\cong[-L,L)$ the appropriate condition is $\varphi\in\sct_\mathrm{b}((-L,L))$ such that the cross-sections as subsets of $\sN B^r$ match smoothly, \iec 
\[
\bigl\{ (\fr(-L)y)_j e_j(-L):\ y\in F\bigr\} =  \bigl\{ (\fr(L)y)_j e_j(L):\ y\in F\bigr\}
\]
and $\varphi'\in \sct_\mathrm{b}([-L,L))$.
While the coordinates $(x,y_1,y_2)$ again provide coordinate vector fields $\partial_x^{\pr}$, $\partial_{y_1}$ and $\partial_{y_2}$, it is more convenient use to the horizontal lift $\partial_x^\sH$ of $\partial_x$ to $M$ instead of $\partial_x^{\pr}$. This lift satifies $\sT\varpi(\partial_x^\sH)=\partial_x^\#$ and is given by
\begin{equation}
\label{eq:xhor}
\partial_x^\sH = \partial_x^{\pr} - (\ln \ell)' y \cdot \nabla_y - \varphi' y \times \nabla_y,\quad \nabla_y=(\partial_{y_1},\partial_{y_2})
\end{equation}
with the intuitive notation for the cross product in $\RR^2$. It incorporates both the effect induced by the variation of the scaling factor (in terms of $\ell'$) and the effect of the twist (in terms of $\varphi'$) along the curve. Finally, the associated family of pullback metrics on $M$ reads
\begin{align*}
 \rs{G^\eps}{(x,y)} & = \eps^{-2}\underbrace{\Bigl(1 -\eps \scpro{\ell(x)\fr(x)y}{\kappa(x)}_{\RR^2}\Bigr)^2\, \dd x^\sH\otimes \dd x^\sH}_{=G_\mathrm{hor}^\eps} \\
 &\hphantom{=}\ + \underbrace{\ell^2(x)\,\bigl(\dd y_1\otimes\dd y_1 + \dd y_2\otimes \dd y_2\bigr)}_{=g_{M_x}}
\end{align*}
with $\dd x^\sH:= \pi_M^*\dd x$.
\item \emph{hollow quantum waveguide, $\dim(F)=1$:} \\
We take the cross-section to be a circle $\mathbb{S}^1=\RR/2\uppi\ZZ$ with an $x$-dependent radius $\ell\in C^\infty_\mathrm{b}(B,[l_-,l_+])$ such that $0<l_- < l_+ < r$. The associated embedding is given by
\[
\varpi :(x,y) \mapsto \ell(x) \bigl(\cos y e_1(x) + \sin y e_2(x)\bigr).
\]
As above, we express the metric  in the basis given by $\partial_y$  and the horizontal lift of $\partial_x$, which in this case is given by
\[
\partial_x^\sH = \partial_x^{\pr} - (\ln\ell)'\,\partial_y.
\]
This eventually leads to the following family of pullback metrics on $M$:
\begin{align*}
 \rs{G^\eps}{(x,y)} &= \eps^{-2} \underbrace{\left[\bigl(1\hspace{-1pt}-\hspace{-1pt}\eps\ell(x)\varkappa(x,y)\bigr)^2 \hspace{-1pt}+ \hspace{-1pt}\bigl(\eps\ell'(x)\bigr)^2\right]  \hspace{-1pt}\dd x^\sH \hspace{-2pt}\otimes\hspace{-1pt}\dd x^\sH}_{=G_\mathrm{hor}^\eps}  + \underbrace{\ell^2(x)\,\dd y \hspace{-2pt}\otimes\hspace{-1pt}\dd y}_{=g_{M_x}},
\end{align*}
where $\varkappa(x,y):=\kappa_1(x)\cos y + \kappa_2(x)\sin y$ is the radial projection of the mean curvature vector.
\end{enumerate}
These examples illustrate the great generality of our approach, \iec even more complex realisations of quantum waveguides may be represented by rather simple embeddings.
\end{Example}

We finally remark that in general the horizontal part of the pullback metric takes the form
\[
G_\mathrm{hor}^\eps(\partial_x^\sH,\partial_x^\sH) = \rho_\eps^2 \stackrel{\eqref{eq:rhoeps}}{=} (1 - \eps \vartheta_\eps)^2,\quad \vartheta_{\eps=0}(x,y) = \scpro{\varpi_x(y)}{\kappa(x)}_{\RR^2},
\]
where we understand $\varpi_x$ as the fibrewise mapping, see Fig.~\ref{fig:embedding},
\begin{equation*}
 \varpi_x: M_x \to N_xB \cong \RR^2
\end{equation*}
From this we immediately infer
\begin{equation}
\label{eq:Ghor0dx}
\expval{G_\mathrm{hor}^\eps}_0(\dd x,\dd x) = 1 + 2 \eps \expval{\vartheta_{\eps=0}}_0 + \cO(\eps^2)
\end{equation}
for the effective metric on covector fields.

\subsection{The Pullback of the Magnetic Potential}

When dealing with quantum tubes, an appropriate modification of the gauge function $\Omega_\eps$ introduced in Lemma~\ref{lem:gauge} leads to a very convenient gauge for the intermediate magnetic potential $\cA_\mathrm{int}^\eps$ (pulled back via $\Phi_\eps$) on $\sN B^r$ in the case of quantum tubes. We will perform all calculations using  (local) coordinates $(x,n_1,n_2)$ -- and therefore a (local) orthonormal frame $\{\tau,e_1,e_2\}$ of $\sT B\oplus \sN B$ -- as introduced in Subsection~\ref{subsec:geom}. Nevertheless, 
 the magnetic potential $\cA_\mathrm{int}^\eps$ is actually a globally well-defined one-form and thus independent of the choice of coordinates.
Let 
\[
\cA_x(x,n) = \cA_{\Phi(x,n)}\bigl(\tau(x)\bigr)\quad,\quad \cA_{n_j}(x,\eps n) = \cA_{\Phi(x,n)}\bigl(e_j(x)\bigr)
\]
for $j\in\{1,2\}$ be the components of the magnetic potential, where $\tau(x)$ and $e_j(x)$ are considered as vectors in $\RR^3$. Then the pulled-back magnetic potential with respect to these coordinates reads
\[
\rs{\Phi_\eps^* \cA}{(x,n)} = \cA_x(x,\eps n)\, \dd x^\# + \eps\bigl( \cA_{n_1}(x,\eps n)\, \dd n_1 + \cA_{n_2}(x,\eps n)\,\dd n_2\bigr).
\]
We set
\[
\Omega_\eps(x,n) := \underbrace{ \sum_{j=1}^2\cA_{n_j}(x,0)n_j}_{\substack{\text{transformation} \\ \text{of Lemma~\ref{lem:gauge}}}} - \frac{\eps^2}{2}\sum_{j,j'=1}^2\frac{\partial \cA_{n_j}(x,0)}{\partial n_{j'}} n_j n_{j'}
\]
which implies
\begin{align*}
\rs{\dd \Omega_\eps}{(x,n)} =& \sum_{j=1}^2 \left[\eps \bigl(\cA_{n_j}(x,0)\bigr)' n_j - \frac{\eps^2}{2} \sum_{j'=1}^2 \left(\frac{\partial\cA_{n_j}(x,0)}{\partial n_{j'}}\right)' n_j n_{j'}\right]\,\dd x^\# \\
& + \sum_{j=1}^2 \left[\eps \cA_{n_j}(x,0) - \frac{\eps^2}{2} \sum_{j'=1}^2\left(\frac{\partial\cA_{n_j}(x,0)}{\partial n_{j'}} +  \frac{\partial\cA_{n_{j'}}(x,0)}{\partial n_j}\right) n_{j'}\right]\dd n_j.
\end{align*}
A Taylor expansion of the coefficients $\cA_x(x,\eps n)$ and $\cA_{n_j}(x,\eps n)$ up to errors of order $\eps^2$ eventually leads to
\begin{align*}
\rs{\cA_\mathrm{int}^\eps}{(x,n)} &:= \rs{(\Phi_\eps^*\cA - \dd\Omega_\eps)}{(x,n)} \nonumber \\
\begin{split}
&\hphantom{:}= \underbrace{\cA_x(x,0)\, \dd x^\#}_{=\pi_{\sN B}^*\cA_B}  + \eps \bigl[\cB_\bot(x) \times n + \cO(\eps^2) \bigr]\, \dd x^\# \\
&\hphantom{:=}\ + \eps^2 \bigl[\tfrac{1}{2}\cB^\bot(x)n + \cO(\eps)\bigr] \times \dd n.
\end{split}
\end{align*}
Here, we introduced the physically relevant magnetic field evaluated on the curve,
$\cB = \dd\cA \in \sct_\mathrm{b}(\rs{\Uplambda^2\sT\RR^3}{\cT^r})$,  which consists of one parallel component
\[
\cB^\parallel(x) := \cB_{c(x)}\bigl(e_1(x),e_2(x)\bigr) = \rs{\left(\frac{\partial \cA_{n_2}}{\partial n_1} - \frac{\partial \cA_{n_1}}{\partial n_2}\right)}{(x,0)}
\]
and two perpendicular components
\begin{align*}
\cB^\bot_1(x) &:=\cB_{c(x)}\bigl(e_2(x),c'(x)\bigr) = \rs{\left(\frac{\partial \cA_x}{\partial n_2} - \frac{\partial \cA_{n_2}}{\partial x}\right)}{(x,0)}, \\
\cB^\bot_2(x) &:=\cB_{c(x)}\bigl(c'(x),e_1(x)\bigr) = \rs{\left(\frac{\partial \cA_{n_1}}{\partial x} - \frac{\partial \cA_x}{\partial n_1}\right)}{(x,0)}
\end{align*}
with respect to the curve $c(B)\subset\RR^3$ (with respect to the local orthonormal frame $\{\tau,e_1,e_2\}$). While the component $\cB^\parallel$ parallel to the curve clearly is unaffected by the choice of transversal directions, both $\cB_\bot\times n$ and $n \times \dd n$ are invariant under a change of a (local) orthonormal frame. Thus, the above expression for $\cA_\mathrm{int}^\eps$ is well defined.

\begin{Example} \label{ex:cAconv}
Let us examine the induced magnetic potential $\cA^\eps$ for the geometric situations illustrated in Example~\ref{ex:geoconv}. Detailed calculations of $\cA^\eps$ in these cases can be found in \cite[Example~5.10]{Haag2016a}.
\begin{enumerate}
\item The horizontal part of $\cA^\eps$ reduces to $\pi_M^*\cA_B + \eps \cA_\sH^\eps$ with
\[
\cA_\sH^{\eps=0}(\partial_x^\sH) = \cB^\bot\times \ell\fr y ,
\]
whereas the vertical part $\eps^2 \cA_\sV^\eps$ is given at leading order by
\[
\cA_\sV^{\eps=0}(\partial_{y_1}) = -\tfrac{1}{2} \cB^\parallel \ell^2 y_2\quad,\quad
\cA_\sV^{\eps=0}(\partial_{y_2}) = \tfrac{1}{2} \cB^\parallel \ell^2 y_1.
\]
Consequently, the magnetic potential of this massive quantum waveguide reads
\[
\cA^\eps = \pi_M^*\cA_B + \eps \bigl(\cB^\bot\times \ell\fr y + \cO(\eps)\bigr)\,\dd x^\sH + \eps^2\Bigl(\tfrac{1}{2} \ell^2 \cB^\parallel y + \cO(\eps)\Bigr)\times \dd y.
\] 
\item Although $\gimel(\partial_x)$ does not vanish in the hollow case, its contribution to the horizontal part $\cA_\sH^\eps$ is of lower order and may thus be neglected, \iec
\[
\cA_\sH^{\eps=0}(\partial_x^\sH) = \cB^\bot \times \begin{psmallmatrix} \ell \cos y \\ \ell \sin y\end{psmallmatrix}.
\]
As far as the vertical part is concerned, we obtain
\[
\cA_\sV^{\eps=0}(\partial_y) = \tfrac{1}{2} \cB^\parallel \begin{psmallmatrix}  \ell \cos y \\ \ell \sin y\end{psmallmatrix} \times \begin{psmallmatrix} - \ell \sin y \\ \ell \cos y\end{psmallmatrix} = \tfrac{1}{2} \cB^\parallel \ell^2
\]
and therefore
\[
\cA^\eps = \pi_M^*\cA_B + \eps \Bigl(\cB^\bot\times \begin{psmallmatrix} \ell \cos y \\ \ell \sin y\end{psmallmatrix}+ \cO(\eps)\Bigr)\,\dd x^\sH + \eps^2\Bigl(\tfrac{1}{2} \ell^2 \cB^\parallel + \cO(\eps)\Bigr)\, \dd y.
\]
\end{enumerate}
\end{Example}

\begin{Remark} \label{rem:cAB}
The intermediate magnetic potential $\cA_\mathrm{int}^\eps$ -- and therefore the resulting pulled-back magnetic potential $\cA^\eps$ -- may be further simplified if one deals with infinite curves ($B=\RR$) with a trivial topology. In this case, we may additionally gauge away the leading part $\pi_{\sN B}^*\cA_B$ by   subtracting the differential of the gauge transformation
\[
(x,n)\mapsto \int_0^x \cA_x(x',0)\ \dd x'.
\]
Here, we set $\rs{\cA}{c(0)}=0$ as reference point without loss of generality. Put differently, the magnetic potential $\cA$ restricted to the infinite curve $c(\RR)\subset\RR^3$ can
be gauged away completely in both the horizontal and vertical directions.
\end{Remark}

\subsection{Moderate Magnetic Fields (\texorpdfstring{$\sigma=0$}{sigma = 0})} \label{sec:unscaled}

We now come to the examination of the adiabatic operator that is associated with quantum tubes in the presence of moderate magnetic fields. Therefore, Theorem~\ref{thm:weak} states that the only possible influence of the magnetic potential arises both from $\cA_B$ and from 
\begin{equation}
\label{eq:cAH}
\rs{\cA_\sH^{\eps=0}(\partial_x^\sH)}{(x,y)} = \rs{\cA_\mathrm{int}^{\eps=0}(\partial_x^\#)}{\varpi(x,y)} = \cB^\bot(x)\times \varpi_x(y)
\end{equation}
averaged against the non-magnetic ground state $\phi_0$, $\cA_1(\partial_x)=\expval{\cA_\sH^{\eps=0}(\partial_x^\sH)}_0$. While the first part $\cA_B$ only depends on the topology of the submanifold $B$ ($\cA_B=0$ for $B=\RR$ due to Remark~\ref{rem:cAB}, but $\cA_B\neq 0$ in general for $B=\mathbb{S}^1$), the second part $\cA_1$ is affected by the concrete modelling of the waveguide (fibres) in terms of the embedding $\varpi$. The following corollary yields easily verifiable conditions   on the geometry which cause a vanishing of the latter potential $\cA_1$:

\begin{Corollary} \label{cor:weak}
Assume that
\begin{enumerate} 
\item massive quantum tubes, $\dim(F)=2$: \\
the map $\varpi_x$ is the restriction of a linear map and that the non-negative ground state $\phi_0(x)$ is centred for all $x\in B$, \iec $\expval{y_1}_0 = 0 = \expval{y_2}_0$,
\item hollow quantum tubes, $\dim(F)=1$: \\
the centre of mass of each cross-section lies in the origin, \iec the integrals
\[
\int_{M_x} \bigl(\varpi_x(y)\bigr)_1\vol_{g_{M_x}}(y)\quad,\quad \int_{M_x} \bigl(\varpi_x(y)\bigr)_2\vol_{g_{M_x}}(y)
\]
vanish for all $x\in B$.
\end{enumerate}
Then the adiabatic operator is given by
\[
\Ha^{\sigma=1} = \eps^2 \underline{\fL} \bigl(\expval{G^\eps_\mathrm{hor}}_0, \cA_B\bigr) +  \lambda_0 + 
\eps^2 V_\mathrm{BH}+ \eps^2 \expval{V _\mathrm{bend}}_0 + \cO(\eps^{2+\alpha})
\]
with errors in $\cL(K_\alpha(B),L^2(B))$ for all $\alpha\in[0,2]$. In particular, if $B=\RR$ is an infinite curve, there are no magnetic effects within the adiabatic operator, \iec
\[
\Ha^{\sigma=0} = \Ha^\mathrm{nm} + \cO(\eps^{2+\alpha}).
\]
\end{Corollary}

\begin{proof}
The vanishing of $\cA_1(\partial_x)=\expval{\cA_\sH^{\eps=0}(\partial_x^\sH)}_0$ reduces by Eq.~\eqref{eq:cAH} to
\[
\expval{\bigl(\varpi_x(y)\bigr)_1}_0 = 0 = \expval{\bigl(\varpi_x(y)\bigr)_2}_0.
\]	
While this trivially reduces to (i) if $\varpi_x$ is linear, we note for (ii) that $\phi_0(x)$ is $y$-independent at leading order by~\eqref{eq:phi0hol} and may thus be extracted from the integral.
\end{proof}

Applying this corollary to the two geometrical situations introduced in Example~\ref{ex:geoconv}, we find the following:

\begin{enumerate}
\item \emph{massive example:} \\
We have $M_x=F$ and $\phi_0(x)=\Phi_0$ for all $x\in B$, where  $\Phi_0$ is the normalised ground state of the Dirichlet Laplacian $-\Delta_y$ on $L^2(F,\dd y)$. Consequently, the condition of Corollary~\ref{cor:weak}(i) reduces to the $x$-independent claim on $\Phi_0$ to be centred, \iec 
\begin{equation}
\label{eq:centredPhi}
\expval{y_1}_{\Phi_0} = 0 = \expval{y_2}_{\Phi_0},
\end{equation}
where $\expval{g(y)}_{\Phi_0} := \scpro{\Phi_0}{g \Phi_0}_{L^2(F)}$ 
for any integrable function $g$ on $F$. One possible realisation thereof is a circular disc $F=\DD^2$ or more generally an ellipse $F=\mathbb{E}^2$.
\item \emph{hollow example:} \\
We have $M_x=\mathbb{S}^1=\RR/2\uppi\ZZ$ for all $x\in B$, and the condition of Corollary~\ref{cor:weak}(ii) is satisfied since
\[
\bigl(\varpi_x(y)\bigr)_1 = \ell(x)\cos y
\]
which implies
\begin{equation}
\label{eq:centredS1}
\int_{M_x} \bigl(\varpi_x(y)\bigr)_1\vol_{g_{M_x}}(y) = \ell^2(x) \int_0^{2\uppi} \cos y\ \dd y = 0,
\end{equation}
and similarly for the other integral. Consequently, if we assume $B=\RR$, the resulting adiabatic operator is independent of any magnetic effects up to errors of order $\eps^{2+\alpha}$ in $\cL(K_\alpha(\RR),L^2(\RR))$. It is discussed in great detail within \cite[Section 5.2]{Haag2015}.
\end{enumerate}

Let us finally compare our adiabatic operator with the effective operator in \cite[Definition~2.5 with $\delta=0$]{Krejcirik2014}. The geometric situation considered there corresponds to that of Example~\ref{ex:geoconv}(i) with $\ell\equiv 1$. We refer to this as \emph{rigid} massive quantum tubes. In particular, this implies that $\lambda_0(x)=\lambda_0$ is constant and we retrieve the case $\alpha=2$, where $\norm{\dd\psi(\partial_x)}_{L^2(B)}=\cO(1)$ for $\psi\in W^2(B)$. Therefore, the operator $\Ha^{\sigma=0}-\lambda_0$ is of order $\eps^2$, hence we may divide it by the same factor and get by Theorem~\ref{thm:weak}:
\begin{align*}
H_\mathrm{rig}^{\sigma=0} &:= \eps^{-2} \bigl(\Ha^{\sigma=0}-\lambda_0\bigr)  =  \underline{\fL} \bigl(\expval{G^\eps_\mathrm{hor}}_0, \cA_\mathrm{eff}^\eps\bigr) +  V_\mathrm{BH}+ \expval{V _\mathrm{bend}}_0 + \cO(\eps^2).
\end{align*}
Its eigenvalues eventually approximate those of the initial magnetic Laplacian $-\Delta_\updelta^\cA$ of the conventional quantum tube $\cT^\eps\subset\RR^3$ up to errors of order $\eps^2$ by Theorem~\ref{thm:main2}. We now compute explicit (local, frame invariant) formulas for the individual terms appearing in $H_\mathrm{rig}^{\sigma=0}$ which itself is a local operator.

\subsubsection*{The Laplacian}

Using the expressions~\eqref{eq:Ghor0dx} and~\eqref{eq:cAH} for the horizontal metric and magnetic potential, respectively, the expansion of the quadratic form of the Laplacian within the above Hamiltonian equals 
\begin{align*}
& \scpro{\psi}{\underline{\fL} \bigl(\expval{G^\eps_\mathrm{hor}}_0, \cA_\mathrm{eff}^\eps\bigr)\psi}_{L^2(B)} \\
&\ = \scpro{\psi}{-\Delta_x^{\cA_B}\psi}_{L^2(B)}   + 2 \eps \biggl( \scpro{\nabla_{\partial_x}^{\cA_B}\psi}{\scpro{\fr\expval{y}_{\Phi_0}}{\kappa}_{\RR^2}\nabla_{\partial_x}^{\cA_B}\psi}_{L^2(B)} \\
&\ \hphantom{=\ + 2 \eps \biggl(} + \Im \scpro{\psi}{\cB^\bot\times\fr\expval{y}_{\Phi_0}\nabla_{\partial_x}^{\cA_B}\psi}_{L^2(B)} \biggr) + \cO\bigl(\eps^2\norm{\psi}_{W^1(B)}^2\bigr).
\end{align*}

\subsubsection*{The Born-Huang Potential}

As far as the non-magnetic Born-Huang potential is concerned, we observe that
\[
G_\mathrm{hor}^\eps(\dd\phi_0,\dd\phi_0) = \bigl(1 + 2 \eps \scpro{\fr y}{\kappa}_{\RR^2}+\cO(\eps^2)\bigr) \bigl(\varphi' (y\times\nabla_y)\Phi_0\bigr)^2
\]
and
\begin{align*}
  \mathrm{div}_g \bigl(\phi_0 G_\mathrm{hor}^\eps(\dd \phi_0, \cdot ) \bigr)  =&\,  \dive_g\Bigl(\Phi_0\bigl(1+2\eps\scpro{\fr y}{\kappa}_{\RR^2}+ \cO(\eps^2)\bigr)(\partial_x^\sH\Phi_0)\,\partial_x^\sH\Bigr) \\
  =  \partial_x^\sH (\Phi_0 \partial_x^\sH \Phi_0) &  + \;2 \eps \Bigl[\varphi' (y\times\nabla_y)\bigl( \scpro{\fr y}{\kappa}_{\RR^2} \Phi_0 \varphi'(y\times\nabla_y)\Phi_0\bigr) \\ 
& - \bigl(\scpro{\fr y}{\kappa}_{\RR^2}\varphi'\bigr)' \Phi_0 (y\times\nabla_y)\Phi_0\Bigr]  + \cO(\eps^2)
\end{align*}
by means of~\eqref{eq:xhor} with $\ell\equiv 1$. In view of Remark~\ref{rem:VBH} the first term doesn't contribute to $V_{\mathrm{BH}}$
and we have 
\begin{align*}
V_\mathrm{BH} &= {\varphi'}^2 \norm{L_y\Phi_0}^2_{L^2(F)} + 2 \eps \expval{ \Bigl[\bigl(\scpro{\fr y}{\kappa}_{\RR^2}\varphi'\bigr)' - L_y \scpro{\fr y}{\kappa}_{\RR^2}\Bigr]L_y}_{\Phi_0} \\
&\hphantom{=}\ + \cO(\eps^2),
\end{align*}
where $L_y:=y\times\nabla_y$ stands for the (skew-symmetric) vertical angular momentum operator. Again, all introduced objects do not depend on the explicitly chosen orthonormal frame.

\subsubsection*{The Bending Potential}

Using the fact that $\rho_\eps = 1 - \eps \scpro{\fr y}{\kappa}_{\RR^2}$, the formulas of \cite[Section 4.3]{Haag2015} lead to
\[
\eps^2\fL(\pi_M^*g_B,0)\ln\rho_\eps = \eps^3 \scpro{\fr y}{\kappa''} + \cO(\eps^4)
\]
and
\begin{align*}
\eps^2\fL(g_\sV,0) \ln\rho_\eps &= g_\sV(\dd\ln\rho_\eps,\dd\ln\rho_\eps) \\
&= \eps^2 \bigl(1+2\eps\scpro{\fr y}{\kappa}_{\RR^2}\bigr)\norm{\kappa}_{\RR^2}^2 + \cO(\eps^4) .
\end{align*}
Thus, the bending potential reads
\[
V_\mathrm{bend} = -\tfrac{1}{4}\norm{\kappa}^2_{\RR^2} - \tfrac{1}{2}\eps\Bigl( \norm{\kappa}^2_{\RR^2}\scpro{\fr y}{\kappa}_{\RR^2} + \scpro{\fr y}{\kappa''}_{\RR^2}\Bigr) + \cO(\eps^2).
\]
and hence
\begin{align*}
\expval{V_\mathrm{bend}}_0 &= -\tfrac{1}{4}\norm{\kappa}^2_{\RR^2} - \tfrac{1}{2}\eps\Bigl( \norm{\kappa}^2_{\RR^2}\scpro{\fr \expval{y}_{\Phi_0}}{\kappa}_{\RR^2} + \scpro{\fr \expval{y}_{\Phi_0}}{\kappa''}_{\RR^2}\Bigr) \\
& \hphantom{=}\ + \cO(\eps^2).
\end{align*}

\subsubsection*{Conclusion}

We conclude that the leading contribution of $H_\mathrm{rig}^{\sigma=0}$ equals
\[
- \Delta^{\cA_B}_x + {\varphi'}^2 \norm{L_y\Phi_0}^2_{L^2(F)} - \tfrac{1}{4}\norm{\kappa}_{\RR^2}^2
\]
which is in accordance with \cite[Theorem~2.6 with $\delta=0$]{Krejcirik2014} for infinite curves. The subsequent order of $H_\mathrm{rig}^{\sigma=0}$ incorporates magnetic effects in terms of $\cA_B$ and $\cB^\bot$. If $\Phi_0$ is centred~\eqref{eq:centredPhi}, this $\cO(\eps)$-contribution to the Hamiltonian reduces to the potential
\[
2\eps\scpro{\Phi_0}{\Bigl[\bigl(\scpro{\fr y}{\kappa}_{\RR^2}\varphi'\bigr)' - L_y \scpro{\fr y}{\kappa}_{\RR^2}\Bigr]L_y \Phi_0}_{L^2(F)}
\]
Here, in particular, we observe that magnetic effects induced by $\cB^\bot$ are not occurrent at this order as stated by Corollary~\ref{cor:weak}(i). Moreover, if in addition the cross-section~$F$ (and, therefore, the non-magnetic ground state $\Phi_0$) is rotationally invariant, it holds that $L_y\Phi_0=0$ and this potential vanishes.

\subsection{Strong Magnetic Fields (\texorpdfstring{$\sigma=1$}{sigma = 1})} \label{sec:scaled}

We continue with the analysis of the adiabatic operator for the framework of quantum tubes with strong magnetic fields. 

\subsubsection{The Case $\alpha=0$}

While hollow quantum waveguides always possess spectrum below this energy regime (more precisely, the constant non-magnetic ground state band in the hollow cases induces spectrum of order $\eps^2$ above the bottom), the regime $\alpha=0$ corresponds to the low-lying part of the spectrum of generic massive quantum waveguides. Therefore, we expand $\Ha^{\sigma=1}$ from Corollary~\ref{cor:strong} for massive quantum tubes as an operator in $\cL(K_0(B),L^2(B))$, \iec on states $\psi\in W^2(B)$ with $\norm{\eps\dd\psi(\partial_x)}_{L^2(B)}=\cO(1)$.

The effective magnetic potential $\cA_\mathrm{eff}^\eps=\cA_B + \eps \cA_1$ is (locally) given by
\[
\rs{\cA_\mathrm{eff}^\eps(\partial_x)}{x} \stackrel{\eqref{eq:cAH}}{=} \underbrace{\cA_{c(x)}\bigl(c'(x)\bigr)}_{=\cA_B(\partial_x)} + \eps \underbrace{\cB^\bot(x)\times \bexpval{\varpi_x(y)}_0}_{=\cA_1(\partial_x)=\expval{\cA_\sH^{\eps=0}(\partial_x^\sH)}_0}.
\]
Moreover, the incorporated metric (evaluated on one-forms) reads
\begin{align*}
\expval{G_\mathrm{hor}^\eps}_\cA(\dd x,\dd x) &= 1 + 2 \eps \expval{\vartheta_{\eps=0}}_0 + \cO(\eps^2) \\
&= 1 + 2 \eps \scpro{\bexpval{\varpi_x(y)}_0}{\kappa}_{\RR^2} + \cO(\eps^2)
\end{align*}
by Equation~\eqref{eq:Ghor0dx} and the comment after Corollary~\ref{cor:strong}. Thus, we get the following expansion for the adiabatic operator:
\begin{align*}
H_\mathrm{a}^{\sigma=1} &  =- \nabla_{\eps\partial_x}^{\eps^{-1}\cA_\mathrm{eff}^\eps} \Bigl(\bigl(1 + 2 \eps \expval{\vartheta_{\eps=0}}_0\bigr)\nabla_{\eps\partial_x}^{\eps^{-1}\cA_\mathrm{eff}^\eps}\cdot\Bigr) + \lambda_0 + \cO(\eps^2) \\
&= - \eps^2 \bigl(1 + 2 \eps \expval{\vartheta_{\eps=0}}_0\bigr)\Delta_x^{\eps^{-1}\cA_\mathrm{eff}^\eps} + \lambda_0+ \cO(\eps^2).
\end{align*}
The relevant quantities herein, that correspond to the geometric setting of Example~\ref{ex:geoconv}(i) and Example~\ref{ex:cAconv}(i), are given by
\begin{equation}
\label{eq:cA1}
\expval{\vartheta_{\eps=0}}_0 = \ell\scpro{\fr \expval{y}_0}{\kappa}_{\RR^2}\quad,\quad \cA_1(\partial_x) = \cB^\bot \times \ell \fr \expval{y}_0
\end{equation}
and thus vanish for a centred non-magnetic ground state $\phi_0$, eventually leading to
\[
\Ha^{\sigma=1} = - \eps^2 \Delta_x^{\eps^{-1}\cA_B} + \lambda_0 + \cO(\eps^2)
\]
is this case.

\subsubsection{The Case $\alpha=2$}\label{sec:652}

We intend to discuss the main magnetic effects within the adiabatic operator if the unperturbed ground state band $\lambda_0(x)=\lambda_0$ is constant. Therefore, in order to apply Theorem~\ref{thm:main2}, we restrict ourselves to infinite curves ($B=\RR$), where $\cA_B=0$. It turns out that the resulting operators are of order $\eps^2$ in $\cL(K_2(\RR),L^2(\RR))$, where $\norm{\dd\psi(\partial_x)}_{L^2(\RR)}=\cO(1)$ for $\psi\in W^2(\RR)$.
In view of Theorem~\ref{thm:strong}, we obtain
\begin{align*}
&  \eps^{-2} \bigl(\Ha^{\sigma=0}-\lambda_0\bigr) \nonumber \\
& \ =\underline{\fL}\bigl(\expval{G^\eps_\mathrm{hor}}_\cA,\cA_P\bigr)
+ \lambda_{0,2}  +  \expval{V_\mathrm{bend}}_\cA + V_\mathrm{BH}^\cA + D_\mathrm{mag} + \cO(\eps^2),
\end{align*}
where $\lambda_{0,2}=\lim_{\eps\to 0} \eps^{-2}(\lambda_\cA-\lambda_0)$ is the second-order correction of the magnetic ground state band to $\lambda_0$. By Theorem~\ref{thm:main2} we could achieve an accuracy of order $\eps^2$ for the approximation of the initial Laplacian~$-\Delta_\updelta^{\eps^{-1}\cA}$ of the conventional quantum tube $\cT^\eps\subset\RR^3$. We, however, focus on the leading contribution of these adiabatic operators and hence on the main magnetic effects.
\begin{itemize}
\item The leading order of the Laplacian is given by $-\Delta_x^{\cA_1}$. In view of~\eqref{eq:Vbend}, the expansion of the averaged bending potential reads
\[
\expval{V_\mathrm{bend}}_\cA = \begin{cases} 
-\tfrac{1}{4}\norm{\kappa}^2_{\RR^2} + \cO(\eps),&\quad \text{mas.}\\
  \cO(\eps),&\quad\text{hol.}
\end{cases}.
\]
\item Standard perturbation theory reveals that (see for example~\cite[Section II – \S\,2]{Kato1980})
\begin{equation}
\label{eq:phi01}
\lambda_{0,2}(x) = \norm{\cA_\sV^{\eps=0} \phi_0}^2_{L^2(\sT^*M_x)} - \scpro{\phi_0}{H^\sV_1 \phi_{0,1}}_{L^2(\sT^*M_x)},
\end{equation}
where
\[
\phi_{0,1} = - (H^\sV_0 - \lambda_0)^{-1} P_0^\bot H^\sV_1 \phi_0
\]
and
\[
H^\sV_1 =  \dd^*\bigl(\I\cA_\sV^{\eps=0}\bigr) + \bigl(\I\cA_\sV^{\eps=0}\bigr)^*\dd.
\]
We mention that $\lambda_{0,2}$ is non-negative due to the diamagnetic inequality \cite{Fournais2010} and may be expressed more concretely in the following sense: Let $\{\zeta_k(x)\}_{k\geq 0}$ be an orthonormal basis of the vertical non-magnetic Laplacian $-\Delta_\sV(x)=-\Delta_{g_{M_x}}$ on~$L^2(M_x)$ with eigenvalues $\{\lambda_k(x)\}_{k\geq 0}$, then one has
\begin{align}
\label{eq:lambda02}
& \lambda_{0,2}  (x) =  \norm{\cA_\sV^{\eps=0}}_{L^2(\sT^*M_x)} \\
  & + \sum_{k>0} \frac{\abs{\scpro{\dd\zeta_k}{\cA_\sV^{\eps=0}\zeta_0}_{L^2(\sT^*M_x)}
- \scpro{\cA_\sV^{\eps=0}\zeta_k}{\dd\zeta_0}_{L^2(\sT^*M_x)}}^2}{\lambda_k - \lambda_0}\nonumber 
\end{align}
We note that, in the hollow case, the corrections to $\lambda_{0,2}$ that arise from $H_0^\sV=-\Delta_\sV+\cO(\eps)$ (inducing an $\cO(\eps)$-deviation between the corresponding eigenfunctions) are of higher order and may thus be neglected.

\item In view of Remark~\ref{rem:VBH}, the expansion of the magnetic Born-Huang potential equals
\begin{align*}
V_\mathrm{BH}^\cA = \norm{\dd\phi_0+\I \cA_{P^\perp}^0}^2_{L^2(\rs{\sH^*M}{M_x})} - \tfrac{1}{2}\bigl(\expval{\hat{\eta}}_0\bigr)' + \cO(\eps),
\end{align*}
where $\cA_{P^\perp}^0:=P_0^\perp \cA_\sH^{\eps=0}\phi_0\in\sct(\sH^*M)\otimes P_0^\perp \cH$ and $\hat{\eta}\in C^\infty(M)$ is the projected mean curvature vector along $\partial_x^\sH$. Herein, we used the fact that $\scpro{\phi_0}{\cA_{P^\perp}^0(\partial_x^\sH)}_{L^2(M_x)}=0$. Furthermore, the first term of this leading order contribution can be rewritten as follows:
\begin{align*}
& \norm{\dd\phi_0 + \I \cA_{P^\perp}^0}_{L^2(\rs{\sH^*M}{M_x})}^2 \\
&\ = \norm{\dd\phi_0}_{L^2(\rs{\sH^*M}{M_x})}^2 + \norm{\cA_\sH^{\eps=0}\phi_0}^2_{L^2(\rs{\sH^*M}{M_x})} - \Bigl(\underbrace{\scpro{ \phi_0}{\cA_\sH^{\eps=0}(\partial_x^\sH)\phi_0}_{L^2(M_x)}}_{=\cA_1(\partial_x)}\Bigr)^2.
\end{align*}
\item The first-order differential operator
\[
D_\mathrm{mag} = - 2\I \Bigl(\underbrace{\scpro{\phi_0}{\cA_{P^\perp}^0(\partial_x^\sH)}_{L^2(M_x)}}_{=0} + \cO(\eps)\Bigr) \nabla_{\partial_x}^{\cA_1}
\]
vanishes at leading order.
\end{itemize}
Consequently, we conclude that
\begin{align} \label{eq:Ha2} 
   \eps^{-2} \bigl(\Ha^{\sigma=0}-\lambda_0\bigr)  =& -\Delta_x^{\cA_1} + \lambda_{0,2}  -\tfrac{1}{4} \norm{\kappa}_{\RR^2}^2 + \norm{\dd\phi_0}_{L^2(\rs{\sH^*M}{M_x})}^2 - \tfrac{1}{2}\bigl(\expval{\hat{\eta}}_0\bigr)' \nonumber \\
& + \norm{\cA_\sH^{\eps=0}\phi_0}^2_{L^2(\rs{\sH^*M}{M_x})}  - \bigl(\cA_1(\partial_x)\bigr)^2  + \cO(\eps)
\end{align}
with errors in $\cL(W^2(\RR),L^2(\RR))$, where the $\kappa$-potential is not present for hollow quantum waveguides.

\subsubsection*{Massive Waveguides}

Let us evaluate the resulting operator~\eqref{eq:Ha2} in more detail for rigid massive quantum tubes ($\eta=0$) as discussed at the end of the previous section, corresponding to the geometric framework of Example~\ref{ex:geoconv}(i) with $\ell\equiv 1$. We recall that  $\cA_\sH^{\eps=0}(\partial_x^\sH)=\cB^\bot\times \fr y$ in this case and thus $\cA_1(\partial_x) = \cB^\bot\times\fr\expval{y}_{\Phi_0}$ by~\eqref{eq:cA1}, where $\Phi_0$ stands for the normalised ground state of the Dirichlet Laplacian $-\Delta_y$ on~$L^2(F,\dd y)$. Consequently, we retrieve with $\partial_x^\sH=\partial_x^{\pr}- \varphi' L_y$:
\begin{align}\label{KRHam}
H_\mathrm{rig}^{\sigma=1} & :=\eps^{-2}\bigl(\Ha^{\sigma=1}-\lambda_0\bigr) \\ \nonumber
&\hphantom{:}= - \Delta_x^{\cB^\bot\times \fr \expval{y}_{\Phi_0}} + \lambda_{0,2} - \tfrac{1}{4}\norm{\kappa}^2_{\RR^2} + {\varphi'}^2 \norm{L_y\Phi_0}^2_{L^2(F)} \\\nonumber
&\hphantom{:=}\ + \norm{(\cB^\bot \times \fr y)\Phi_0}^2_{L^2(F)} - \bigl(\cB^\bot\times \fr\expval{y}_{\Phi_0}\bigr)^2 + \cO(\eps) \\\nonumber
&\hphantom{:}= \scpro{\Phi_0}{-\Delta_\sH^{\cB\times\fr y}(\cdot\Phi_0)}_{L^2(F)} + \lambda_{0,2} - \tfrac{1}{4}\norm{\kappa}^2_{\RR^2} + \cO(\eps).
\end{align}
The last transformation follows from a straightforward calculation, using the fact that
\[
0 = \dd\Bigl(\underbrace{\norm{\Phi_0}^2_{L^2(F)}}_{=1}\Bigr) = 2 \scpro{\Phi_0}{\dd\Phi_0(\partial_x^\sH)}_{L^2(F)} = 2 \expval{L_y}_{\Phi_0}.
\]
This operator exactly coincides with the effective operator in \cite[Definition~2.5 with $\delta=1$]{Krejcirik2014}. We also find the same representation for the correction $\lambda_{0,2}$ of the ground state band:
\[
\lambda_{0,2} \stackrel{eqref{eq:phi01}}{=} {\cB^\parallel}^2 \left(\tfrac{1}{4}\expval{\norm{y}_{\RR^2}^2}_{\Phi_0} - \scpro{\Phi_0}{-\I L_y \Phi_{0,1}}_{L^2(F)}\right).
\]
We eventually note that this operator reduces to
\begin{align*}
H_\mathrm{rig}^{\sigma=1} & = - \partial^2_x + {\varphi'}^2 \norm{L_y\Phi_0}^2_{L^2(F)} - \tfrac{1}{4}\norm{\kappa}^2_{\RR^2} \\
&\hphantom{=}\ + {\cB^\parallel}^2 \left(\tfrac{1}{4}\expval{\norm{y}_{\RR^2}^2}_{\Phi_0} - \scpro{\Phi_0}{-\I L_y \Phi_{0,1}}_{L^2(F)}\right) \\
&\hphantom{=}\ + \norm{(\cB^\bot \times \fr y)\Phi_0}^2_{L^2(F)}  + \cO(\eps)
\end{align*}
if $\Phi_0$ is centred~\eqref{eq:centredPhi}. We finally consider the case where $F$ is rotationally invariant and without loss of generality $\varphi\equiv 0$ (\iec $\fr=\id_{2\times 2}$). This yields that $\Phi_0$ is rotationally invariant and, in particular, centred. The rotational invariance implies $L_y\Phi_0=0$ and thus
\begin{align*}
\norm{(\cB^\bot\times y)\Phi_0}^2_{L^2(F)} &= (\cB_1^\bot)^2 \expval{y_2^2}_{\Phi_0} + (\cB_2^\bot)^2 \expval{y_1^2}_{\Phi_0} - 2 \cB_1^\bot \cB_2^\bot \underbrace{\expval{y_1 y_2}_{\Phi_0}}_{=0} \\
&= \tfrac{1}{2} \norm{\cB^\bot}^2_{\RR^2} \expval{\norm{y}_{\RR^2}^2}_{\Phi_0}.
\end{align*}
Therefore, using the fact that
\[
\scpro{\Phi_0}{-\I L_y \Phi_{0,1}}_{L^2(F)} = \scpro{-\I L_y\Phi_0}{\Phi_{0,1}}_{L^2(F)} = 0
\]
due to integration by parts and $\rs{\Phi_0}{\partial F}=0$, the corresponding adiabatic operator furthermore reduces to
\[
H_\mathrm{rig}^{\sigma=1} = -\partial_x^2 - \tfrac{1}{4}\norm{\kappa}^2_{\RR^2} + \tfrac{1}{4}\expval{\norm{y}^2_{\RR^2}}_{\Phi_0}\bigl({\cB^\parallel}^2 + 2 \norm{\cB^\bot}^2_{\RR^2}\bigr) + \cO(\eps).
\]

\begin{Example}
If we assume that $F$ is the circular disc with radius $\varrho$, a small calculation shows that
\[
\expval{\norm{y}^2_{\RR^2}}_{\Phi_0} = \frac{j_{01}^2-2}{3j_{01}^2} \varrho^2 \approx 0.218 \varrho^2,
\]
where $j_{01}\approx 2.405$ denotes the first zero of the zeroth Bessel function $\cJ_0$ of first kind.
\end{Example}

\subsubsection*{Hollow Waveguides}

We continue the analysis of~\eqref{eq:Ha2} for hollow quantum tubes with $F=\mathbb{S}^1=\RR/2\uppi\ZZ$ as introduced in Example~\ref{ex:geoconv}(ii). We first note that
\[
\phi_0(x) \stackrel{\eqref{eq:phi0hol}}{=} \Vol_{g_{M_x}}(\mathbb{S}^1) + \cO(\eps) = \tfrac{1}{\sqrt{2\uppi\ell(x)}}+\cO(\eps),\quad g_{M_x} =\ell^2(x) \dd y\otimes\dd y
\]
which entails $\dd\phi_0(\partial_x^\sH)=- \tfrac{1}{2}(\ln \ell)'\phi_0 + \cO(\eps)$. This together with  $\hat{\eta} = -(\ln\ell)'$ then gives
\begin{align*}
\norm{\dd\phi_0}_{L^2(\rs{\sH^*M}{M_x})}^2 - \tfrac{1}{2}\bigl(\expval{\hat{\eta}}_0\bigr)' & = \tfrac{1}{4} {(\ln\ell)'}^2 + \tfrac{1}{2} (\ln\ell)'' + \cO(\eps) \\
& = \tfrac{1}{2}\tfrac{\ell''}{\ell} - \tfrac{1}{4} \left(\tfrac{\ell'}{\ell}\right)^2 + \cO(\eps).
\end{align*}
Considering the magnetic terms, we have $\cA_1(\partial_x)=\cO(\eps)$ due to~\eqref{eq:centredS1}. Moreover,
\begin{align*}
\norm{\cA_\sH^{\eps=0}(\partial_x^\sH)\phi_0}^2_{L^2(M_x)} &= \frac{1}{2\uppi}\int_0^{2\uppi} \Bigl[\cB^\bot\times \begin{psmallmatrix} \ell \cos y \\ \ell \sin y\end{psmallmatrix}\Bigr]^2\ \dd y + \cO(\eps) \\
&= \tfrac{1}{2}\ell^2\norm{\cB^\bot}^2_{\RR^2} + \cO(\eps).
\end{align*}
It remains to evaluate the correction $\lambda_{0,2}$ using~\eqref{eq:lambda02}. To do so, we introduce the basis
\[
\bigl\{\zeta_0^\mathrm{e},\{\zeta_k^\mathrm{e},\phi_k^\mathrm{o}\}_{k\geq 1}\bigr\}
\]
of orthonormal eigenfunctions of $-\Delta_\sV=-\ell^{-2}\,\partial_y^2$ on $L^2(\mathbb{S}^1,\ell\, \dd y)$, where
$\zeta_0^\mathrm{e} = \tfrac{1}{\sqrt{2\uppi\ell}}$ and
\[
\zeta_k^\mathrm{e} = \tfrac{1}{\sqrt{\uppi\ell}} \cos(ky)\quad,\quad\zeta_k^\mathrm{o} = \tfrac{1}{\sqrt{\uppi\ell}} \sin(ky)
\]
with eigenvalues $\lambda_k=k^2/\ell^2$.
The first term of~\eqref{eq:lambda02} then equals
\[
\int_0^{2\uppi} \tfrac{1}{\ell^2} \Bigl(\tfrac{1}{2}\ell^2 \cB^\parallel \zeta_0^\mathrm{e}\Bigr)^2\ \ell\,\dd y = \tfrac{1}{4}{\cB^\parallel}^2 \underbrace{\scpro{\zeta_0^\mathrm{e}}{\zeta_0^\mathrm{e}}_{L^2(\mathbb{S}^1,\ell\,\dd y)}}_{=1},
\]
while the second one reads
\begin{align*}
&\sum_{k\geq 1} \sum_{\bullet\in\{\mathrm{e,o}\}} \frac{\babs{\int_0^{2\uppi}\tfrac{1}{\ell^2}(\partial_y \zeta_k^\bullet)\left(\tfrac{1}{2}\ell^2\cB^\parallel\zeta_0^\mathrm{e}\right)- \tfrac{1}{\ell^2}\left(\tfrac{1}{2}\ell^2\cB^\parallel\zeta_k^\bullet\right)\overbrace{(\partial_y \zeta_0^\mathrm{e})}^{=0}\ \ell\,\dd y}^2}{k^2/\ell^2} \\
&\ = \sum_{k\geq 1} \tfrac{1}{2}\tfrac{\ell^2}{k^2}{\cB^\parallel}^2 \left(\babs{\scpro{\partial_y \zeta_k^\mathrm{e}}{\zeta_0^\mathrm{e}}_{L^2(\mathbb{S}^1,\ell\,\dd y)}}^2 - \babs{\scpro{\partial_y \zeta_k^\mathrm{o}}{\zeta_0^\mathrm{e}}_{L^2(\mathbb{S}^1,\ell\,\dd y)}}^2\right) \\
&\ = \sum_{k\geq 1} \tfrac{1}{2}\ell^2{\cB^\parallel}^2 \Bigl(\babs{\underbrace{\scpro{ \zeta_k^\mathrm{o}}{\zeta_0^\mathrm{e}}_{L^2(\mathbb{S}^1,\ell\,\dd y)}}_{\text{$=0$ for all $k\geq 1$}}}^2 - \babs{\underbrace{\scpro{\zeta_k^\mathrm{e}}{\zeta_0^\mathrm{e}}_{L^2(\mathbb{S}^1,\ell\,\dd y)}}_{\text{$=0$ for all $k\geq 1$}}}^2\Bigr) \\
&\ = 0.
\end{align*}
Finally, we gather all terms of~\eqref{eq:Ha2} and find
\begin{align*}
H_\mathrm{hol}^{\sigma=1} &= \eps^{-2}\bigl(\Ha^{\sigma=1}-\lambda_0\bigr) \\
&= - \partial_x^2 + \tfrac{1}{2}\tfrac{\ell''}{\ell} - \tfrac{1}{4} \left(\tfrac{\ell'}{\ell}\right)^2 +\tfrac{1}{4}\ell^2 \bigl({\cB^\parallel}^2 + 2 \norm{\cB^\bot}^2_{\RR^2}\bigr) + \cO(\eps).
\end{align*}
We observe a similar magnetic contribution, proportional to ${\cB^\parallel}^2 + 2 \norm{\cB^\bot}^2_{\RR^2}$, as in the case of rotationally invariant massive quantum tubes above. We finally mention that the non-magnetic part
\[
- \partial_x^2 + \tfrac{1}{2}\tfrac{\ell''}{\ell} - \tfrac{1}{4} \left(\tfrac{\ell'}{\ell}\right)^2 
\]
of $H_\mathrm{hol}^{\sigma=1}$ was already found in \cite[Equation~(45)]{Haag2015}.

\section*{Acknowledgments}

S.~Haag was supported by the German Research Foundation (DFG) within the Research Training Group 1838 on \enquote{Spectral theory and dynamics of quantum systems}.


\begin{thebibliography}{KRT15}

\bibitem[Bis75]{Bishop1975}
R.~L. Bishop, \emph{There is more than one way to frame a curve}, Amer. Math.
  Monthly \textbf{82} (1975), no.~3, 246--251.

\bibitem[EK05]{Ekholm2005}
T.\ Ekholm and H.\ Kova\v{r}{\'i}k, \emph{Stability of the magnetic
  {S}chr\"{o}dinger operator in a waveguide}, Comm.\ Part.\ Diff.\ Eq.
  \textbf{30} (2005), no.~4, 539--565.

\bibitem[EK15]{Exner2015}
P.\ Exner and H.\ Kova{\v{r}}{\'\i}k, \emph{Quantum waveguides}, Springer,
  2015.

\bibitem[FH10]{Fournais2010}
S.\ Fournais and B.\ Helffer, \emph{Spectral methods in surface
  superconductivity}, Progress in Nonlinear Differential Equations and Their
  Applications 77, Birkh\"{a}user-Verlag, 2010.

\bibitem[Gru08]{Grushin2008}
V.\ Grushin, \emph{Asymptotic behavior of the eigenvalues of the
  {S}chr\"{o}dinger operator in thin closed tubes}, Mathematical Notes
  \textbf{83} (2008), no.~3, 463--477.

\bibitem[Haa16]{Haag2016a}
S.\ Haag, \emph{The adiabatic limit of the connection {L}aplacian with
  applications to quantum waveguides}, Ph.D. thesis, {E}berhard {K}arls
  {U}niversit\"{a}t {T}\"{u}bingen, 2016.

\bibitem[HL18]{Haag2016b}
S.\ Haag and J.\ Lampart, \emph{The adiabatic limit of the connection
  {L}aplacian}, The Journal of Geometric Analysis, \textbf{29} (2019), no.~3, 2644--2673.

\bibitem[HLT15]{Haag2015}
S.\ Haag, J.\ Lampart, and S.\ Teufel, \emph{Generalised quantum waveguides},
  Ann.\ Henri Poincar\'{e} \textbf{16} (2015), no.~11, 2535--2568.

\bibitem[Kat80]{Kato1980}
T.\ Kat\={o}, \emph{Perturbation theory for linear operators}, Springer-Verlag,
  1980.

\bibitem[KR14]{Krejcirik2014}
D.\ Krej\v{c}i\v{r}\'{i}k and N.\ Raymond, \emph{Magnetic effects in curved
  quantum waveguides}, Ann.\ Henri Poincar\'{e} \textbf{15} (2014), no.~10,
  1993--2024.

\bibitem[KRT15]{Krejcirik2015}
D.\ Krej\v{c}i\v{r}\'{i}k, N.\ Raymond, and M.\ Tu\v{s}ek, \emph{The magnetic
  {L}aplacian in shrinking tubular neighbourhoods of hypersurfaces}, The
  Journal of Geometric Analysis \textbf{25} (2015), no.~4, 2546--2564.

\bibitem[LT17]{Lampart2016}
J.\ Lampart and S.\ Teufel, \emph{The adiabatic limit of {S}chr{\"o}dinger
  operators on fibre bundles}, Mathematische Annalen \textbf{367} (2017),
  no.~3--4, 1647--1683.
  
\bibitem[Mar07]{Martinez2007}
 A.\ Martinez, \emph{A general effective Hamiltonian method},
 Rendiconti Lincei-Matematica e Applicazioni \textbf{18} (2007), no.~3, 269--277.
 
  
\bibitem[Sch96]{Schick1996}
T.\ Schick, \emph{Analysis of $\partial$-manifolds of bounded geometry,
  {H}odge-de {R}ham isomorphism and $l^2$-index theorem}, Ph.D. thesis,
  Johannes Gutenberg Universit\"{a}t Mainz, 1996.

\bibitem[Shu92]{Shubin1992}
M.A.\ Shubin, \emph{Spectral theory of elliptic operators on non-compact
  manifolds}, Ast\'{e}risque \textbf{1} (1992), no.~207, 35--108.

\bibitem[Sim83]{Simon1983}
B.\ Simon, \emph{Semiclassical analysis of low lying eigenvalues. I.
  nondegenerate minima: Asymptotic expansions}, Annales de l'Institut Henri
  Poincar\'{e}, Section A \textbf{38} (1983), no.~3, 295--308.

\bibitem[WT13]{WaTe2013}
J.\ Wachsmuth and S.\ Teufel, \emph{Effective {H}amiltonians for constrained
  quantum systems}, Memoirs of the American Mathematical Society (2013),
  no.~1083.

\end{thebibliography}
\end{document}